\documentclass[journal]{IEEEtran}
\usepackage{latexsym}
\usepackage{graphicx}
\usepackage{amsfonts,amssymb,amsmath}

\usepackage[
    backend=biber,
    style=ieee,
  ]{biblatex}

\def\tsc#1{\csdef{#1}{\textsc{\lowercase{#1}}\xspace}}
\tsc{WGM}
\tsc{QE}
\tsc{EP}
\tsc{PMS}
\tsc{BEC}
\tsc{DE}

\usepackage{caption}
\usepackage{subcaption}
\usepackage{comment}
\usepackage{setspace}
\usepackage{tabularx}
\usepackage{makecell}
\usepackage{amsmath}
\usepackage{amsthm}
\usepackage{multirow}
\usepackage{subcaption}
\usepackage{graphics}
\usepackage[algo2e, ruled, vlined, linesnumbered]{algorithm2e} 
\usepackage{xcolor}

\usepackage{booktabs} 
\usepackage{tablefootnote} 

\usepackage{mathtools}
\newcommand{\defeq}{\vcentcolon=}

\newcommand{\norm}[1]{\left\lVert#1\right\rVert}

\let\oldeqref\eqref
\renewcommand*{\eqref}[1]{Eq.~\oldeqref{#1}}

\usepackage{amssymb}
\usepackage{comment}
\usepackage{pifont}
\usepackage{url}
\usepackage{tabularx}
    \newcolumntype{L}{>{\raggedright\arraybackslash}X}

\ifodd 0
    \newcommand\revt[1]{{\color{blue}#1}}
    \newcommand{\com}[1]{\textbf{\color{red} (COMMENT: #1)}} 
\else
    \newcommand\revt[1]{{#1}}
    \newcommand{\com}[1]{}
\fi

\ifodd 0
    \newcommand\rev[1]{{\color{blue}#1}}
    
    \newcommand{\com}[1]{\textbf{\color{red} (COMMENT: #1)}} 
\else
    \newcommand\rev[1]{{#1}}

\fi

\ifodd 0
    \newcommand\revrr[1]{{\color{blue}#1}}
    \newcommand{\comrr}[1]{\textbf{\color{red} (COMMENT: #1)}} 
\else
    \newcommand\revrr[1]{{#1}}
    \newcommand{\comrr}[1]{}
\fi
\newcommand{\xmark}{\ding{55}}

\DeclareMathOperator*{\argmin}{arg\,min}

\newtheorem{theorem}{Theorem}

\newtheorem{lemma}{Lemma}

\usepackage{hyperref}
\def\BibTeX{{\rm B\kern-.05em{\sc i\kern-.025em b}\kern-.08em T\kern-.1667em\lower.7ex\hbox{E}\kern-.125emX}}

\begin{document}

\title{Optimizing Two-Truck Platooning with Deadlines}

\author{Wenjie~Xu,
        Titing~Cui,
        and~Minghua~Chen,~\IEEEmembership{Fellow,~IEEE}%
\thanks{
The work presented in this paper was supported in part by a General Research Fund from Research Grants Council, Hong Kong (Project No. 14207520), an InnoHK initiative, The Government of the HKSAR, and Laboratory for AI-Powered Financial Technologies.
}
\thanks{W. Xu is with \'Ecole Polytechnique F\'ed\'erale de Lausanne, Lausanne CH-1015, Switzerland.}
\thanks{T. Cui is with Joseph M. Katz Graduate School of Business, University of Pittsburgh, Pittsburgh, PA 15260, USA.}
\thanks{M. Chen is with School of Data Science, City University of Hong Kong, 83 Tat Chee Avenue, Kowloon Tong, Kowloon, Hong Kong.}
\thanks{The first two authors are co-primary authors. A part of the work was done when all the three authors were with The Chinese University of Hong Kong. Corresponding author: Minghua Chen (minghua.chen@cityu.edu.hk).}
}

\markboth{}
{Xu \MakeLowercase{\textit{et al.}}: Bare Demo of IEEEtran.cls for IEEE Journals}

\maketitle

\begin{abstract}
\rev{We study a transportation problem where two heavy-duty trucks travel across the national highway from separate origins to destinations, subject to individual deadline constraints.} Our objective is to minimize their total fuel consumption by jointly optimizing path planning, speed planning, and platooning configuration. Such a two-truck platooning problem is pervasive in practice yet challenging to solve due to \revt{hard deadline constraints and enormous platooning configurations to consider.}
We first leverage a unique problem structure to significantly simplify platooning optimization and present a novel formulation. We prove that the two-truck platooning problem is \textit{weakly} NP-hard and admits a Fully Polynomial Time Approximation Scheme (FPTAS). The FPTAS can achieve \revt{a fuel consumption within a ratio of $(1+\epsilon)$ to the optimal (for any $\epsilon>0$) with a time complexity polynomial in the size of the transportation network and $1/\epsilon$.} These results are in sharp contrast to the general multi-truck platooning problem, which is known to be APX-hard and repels any FPTAS. As the FPTAS still incurs excessive running time for large-scale cases, we design an efficient dual-subgradient algorithm for solving large-/national- scale instances. It is an iterative algorithm that always converges. We prove that each iteration only incurs polynomial-time complexity, albeit it requires solving an integer linear programming problem optimally. We characterize a condition under which the algorithm generates an optimal solution and \revt{derive} a posterior performance bound when the condition is not met. Extensive simulations based on real-world traces show that our joint solution of path planning,  speed planning, and platooning saves up to $24\%$ fuel as compared to baseline alternatives.
\end{abstract}



\IEEEpeerreviewmaketitle

\section{Introduction}
Heavy-duty trucks, transporting freights across national highways, play a critical role in society's economic development. It is reported that the U.S. trucking industry hauled around $80\%$ of all freight tonnage and brought in 792 billion U.S. dollars as revenue in 2019~\cite{ATA}.
This number would rank 18th if measured against the country's GDP~\cite{GDPranking2019}. Meanwhile, the {trucking industry} is also responsible for a large {fraction of the world's energy consumption and greenhouse gas emissions. {For example}, in the U.S., {with} only 4\% of the vehicle population, heavy-duty trucks account for $18\%$ energy consumption of the entire transportation sector and $5\%$ of the greenhouse gas emission~\cite{ATA}. In addition, fuel cost corresponds to $24\%$ of the total operational cost in the trucking industry in 2019~\cite{OperCost2020analysis}.} These alerting observations, {together with that the global freight activity is predicted to increase by a factor of 2.4 by 2050~\cite{futuretruck},} make it critical to develop effective solutions for improving fuel efficiency of truck transportation.

Truck platooning, where two or more trucks travel in convoy with small headway, is a promising fuel-saving approach. It is made possible in practice by the novel semi-automated driving technology, also referred to as the Cooperative Adaptive Cruise Control (CACC) system~\cite{Alam2013LCC, Besselink2016CyberP, Turri2017CooperativeL}. Similar to what racing cyclists exploit, trucks traveling in a platoon experience a reduction in air drag, which translates into lower fuel consumption. Many studies, e.g.,~\cite{2000-01-3056, 2014-01-2438, Alam2015Sustainable, AeroTruck, Tsugawa2016}, suggest up to 20\% fuel saving for the non-leading vehicles in a platoon. Meanwhile, platooning also allows vehicles to drive closer than otherwise at the same speed, improving traffic throughput~\cite{vanArem2006} and even safety. To date, platooning has received significant attention from private fleets and commercial carriers~\cite{tsugawa2011automated}.

However, it is algorithmically non-trivial to capitalize the platooning potential to minimize the convoy's total fuel consumption. First, it requires exploring three distinct design spaces: path planning, speed planning, and platooning configuration. Path planning and speed planning are well-recognized mechanisms for saving fuels by optimizing the driving paths and the speed profiles over individual road segments. Real-world studies show up to 21\% fuel saving by driving at fuel-economic speeds along an optimized path~\cite{Tunnell2011EstimatingTF}. Meanwhile, \revt{platooning configurations define the road segments for platooning and the trucks to participate in each platoon. In other words, it describes where and when a set of the trucks form and end a platoon.} It adds another layer of combinatorial structure to the design space.  

Second, the problem becomes even more complicated when we consider transportation deadlines for individual trucks. Freight delivery with time guarantee is common in truck operation; see examples and discussions in~\cite{amazon,ashby2006protecting}. As estimated by the U.S. Federal Highway Administration (FHWA)~\cite{mallett2006freight}, an unexpected delay can increase freight cost by 50\% to 250\%. The operator needs to keep track of the time spent so far and the remaining time budget for individual trucks, when optimizing {over the design space discussed above}, to ensure on-time arrival at destinations. Such a 4-way coupling among path planning, speed planning, platooning configuration, and deadline creates a unique challenge. Indeed, the problem of minimizing fuel consumption for multiple platooning trucks is shown to be APX-hard in general~\cite{DRAKE200415}, and it admits no polynomial-time algorithm with a fuel consumption within a constant ratio to the minimum, unless $\mathbf{P}=\mathbf{NP}$~\cite{Larson2016CoordinatedPR, Vandehoef2015}.

\begin{table*}
    \caption{Summary of existing studies on minimizing fuel consumption for truck transportation with platooning. 
    }
    \label{tab:Platooning_research}
    
    {\centering
    \begin{tabular}{|l|c|c|c|c|l|c|}
    \hline
     \textbf{Existing Study} & \textbf{Path Planning} & \textbf{Speed Planning} & \textbf{Platooning} & \textbf{Deadline} & \textbf{Approach} & {\textbf{Bound$^*$}}\\
     
    \hline ~\cite{Vandehoef2015,Vandehoef201509, Liang2016, Vandehoef2018} &\xmark  &\checkmark & \xmark &  \checkmark & Convex optimization 
    & \xmark\\
    \hline
    ~\cite{Larson2013} & \xmark & \checkmark & \checkmark& \checkmark & Distributed method& \xmark\\
    \hline
        ~\cite{LARSSON2015258} & \checkmark &\xmark & \checkmark & \xmark& Hub heuristic &\xmark \\
    \hline
        ~\cite{Larson2016CoordinatedPR} & \checkmark &\xmark & \checkmark &\checkmark & Integer programming (IP) & \xmark\\
    \hline
        ~\cite{LUO2018213} & \checkmark & \revt{\checkmark} & \checkmark &\checkmark & K-set clustering + IP &\xmark \\
    \hline
    ~\cite{BOYSEN201826} & \xmark & \xmark & \checkmark & \checkmark & Mixed integer programming & \checkmark\\
    
    \hline
    \revt{~\cite{Luo2021repeated}} & \checkmark & \xmark & \checkmark & \checkmark & Iterative heuristic & \xmark \\
    \hline
        This work & \checkmark& \checkmark & \checkmark & \checkmark & Dual-based optimization & Posterior \\
    \hline
    \end{tabular}}
    
    \vspace{2pt}
    \revt{$^*$: Entries in the column represent whether there is a guaranteed performance gap to the optimal. Usually, performance bounds are independent to the solutions obtained by the algorithm and can be computed beforehand. Meanwhile, posterior performance pounds are solution dependent and can be computed after the solution is obtained.
    }
\end{table*}

\textbf{Existing studies}. Previous researches on truck platooning optimization \revt{mostly} consider either speed planning for given paths or path planning with \revt{a set of constant speeds}. Assumed all the trucks drive on their shortest or pre-specified paths, the authors of \cite{Vandehoef2015, Vandehoef201509, Vandehoef2018} obtain the optimal speed profiles for each truck by solving a series of convex optimization problems. The authors in \cite{Liang2016} and \cite{ZHANG20171} use a similar technique to solve a fixed-path platooning problem. When the traveling speed on each road segment is fixed, the path planning in the platooning problem is often formulated as an integer linear programming~(ILP) problem  \cite{Larson2016CoordinatedPR, LUO2018213, LARSEN2019249, Nourmohammadzadeh2016TheFP, Nourmohammadzadeh2018}. In \cite{Larson2016CoordinatedPR, LUO2018213, LARSEN2019249}, heuristics such as the best-pair algorithms are proposed to solve the ILP in the platooning problem. Meanwhile, \cite{Nourmohammadzadeh2016TheFP} and \cite{Nourmohammadzadeh2018} apply genetic algorithms to path planning. \cite{BOYSEN201826} \revt{employs} a mixed-integer programming model to study the platooning of trucks along an identical path but with different time windows and maximum platooning length.
\revt{The studies in \cite{LUO2018213, Nourmohammadzadeh2018fuel} consider the platooning optimization with discrete speeds for individual edges, which can be modeled as a path planning problem by introducing multiple parallel edges between the starting and ending nodes of original edges, each with a constant (and different) speed.}

\revrr{Other related studies include~\cite{XIONG2021482,LEE2021102664,SALA2021116,ZHOU2021102882,Luo2021repeated}.}
\cite{XIONG2021482} \revt{formulates} a Markov decision process to optimize coordination of vehicle platooning at highway junctions.
\cite{LEE2021102664} \revt{presents} a new platoon formation strategy that optimizes the platooning number and configuration of heterogeneous vehicles in a single route.
\cite{SALA2021116} \revt{proposes} a macroscopic model to analyze the probability distribution of the length of platoons,  for a given traffic demand.
\cite{ZHOU2021102882} \revt{develops} an analytical model to investigate the impacts of platoon size.
\revt{\cite{Luo2021repeated} develops a heuristic iterative procedure to optimize the routes and departure schedules for the general k-platooning problem under deadline constraints, without speed planning involved. The heuristic iterative procedure may not converge or have performance guarantee. In comparison, in this paper, we explore the full design space, to jointly optimize continuous speed planning, path planning, and platooning configuration, for a popular two-truck platooning problem. The iterative procedure in our scheme is a dual subgradient algorithm that guarantees to converge to the dual optimal, and we provide a posterior performance bound for it. We refer readers to \cite{BHOOPALAM2018212} for a comprehensive overview of studies related to coordinated platooning optimization.}

Table \ref{tab:Platooning_research} summarizes existing studies on fuel consumption {minimization} in truck platooning.

\textbf{Contributions.} In this paper, we focus on a prevalent two-truck platooning scenario where two heavy-duty trucks travel across the national highway from separate origins to destinations~\cite{UKreport2014, Floridareport2018, NACFEReport2016, Janssen2015TruckPD, 2017StudyOC}. Two-truck platooning is not only the mainstream of current platooning practice~\cite{NACFEReport2016}, but it also incurs minimum safety concern for the surrounding traffic sharing the roads~\cite{UKreport2014}~\cite{Janssen2015TruckPD}~\cite{Floridareport2018}. 

To our best knowledge, this is the first work in the literature to minimize total fuel consumption for truck transportation by simultaneously optimizing path planning, speed planning, and platooning configuration, subject to individual deadline constraints. We also consider \rev{departure coordination},
where trucks can strategically wait at the origins for proper schedule alignment to form an efficient platoon later. We make the following contributions.

$\rhd$ In Sec.~\ref{sec:model}, we present a novel formulation for the two-truck platooning problem of minimizing fuel consumption subject to individual deadline constraints, taking into account the entire design space of path planning, speed planning, and platooning configuration. Our formulation is built upon an elegant problem structure that significantly simplifies the platooning optimization. \revt{We leverage the formulation to prove that the two-truck platooning problem is \emph{weakly} NP-hard and admits a Fully Polynomial Time Approximation Scheme (FPTAS). It guarantees to return an approximate and feasible solution, whose fuel consumption is no large than $(1+\epsilon)$ times the minimum for any $\epsilon>0$, with a time complexity polynomial in the size of the transportation network and $1/\epsilon$.} These results are in sharp contrast to the general multi-truck platooning problem that is APX-hard and repels any FPTAS~\cite{Larson2016CoordinatedPR, Vandehoef2015}, suggesting that the two problems are fundamentally different.

$\rhd$ While the FPTAS works effectively for small-/medium- scale problem instances, it may still incur excessive running time for large-scale problem instances in practice. We thus develop an efficient dual-based algorithm for national-scale problem instances in Sec.~\ref{sec:meth}. {It is an iterative algorithm that always converges. Further, we prove that each iteration only incurs polynomial time complexity, albeit it requires solving an integer linear programming problem optimally. 
    We characterize a condition under which the algorithm generates an optimal solution and derive a posterior performance bound when the condition is not satisfied.}

$\rhd$ In Sec.~\ref{sec:simulation}, we carry out extensive numerical experiments using real-world traces over the U.S. national highway system. \revt{The results show that our algorithm obtains close-to-minimum fuel consumption for the two-truck platooning problem} and achieves up to $24\%$ fuel saving as compared to baseline alternatives adapted from state-of-the-art schemes. 

We conclude and discuss future work in Sec.~\ref{sec:conclusion}.

\section{Model and Problem Formulation}
\label{sec:model} 

\subsection{System Model}

We model a national highway network as a directed graph $G\triangleq(V,E)$.
Each edge $e\in E$ represents a road segment with homogeneous grade,
surface type, and environmental conditions. Each node $v\in V$ denotes
an intersection or connecting point of adjacent road segments. We
define $N\triangleq|V|$ as the number of nodes and $M\triangleq|E|$
as the number of edges. For each $e\in E$, we use $D_{e}>0$ to denote
its length and $r_{e}^{l}>0$ (resp. $r_{e}^{u}$) to denote the minimum
(resp. maximum) driving speed. {

There are mainly three types of truck fuel consumption models~\cite{zhou2016review}: (i) first-principle white-box models, e.g., the kinematics model based on the Newton's second law~\cite{cachon2007fuel,heywood2018internal, rakha2011virginia, moskwa1992modeling,Liang2016, Vandehoef2018, Guo2019}; (ii) black-box models based on fitting a mathematical function to real-world data, e.g.,~\cite{saerens2013assessment, hellstrom2009look, park2010development, post1984fuel, faris2011vehicle, piccoli2013estimating, leung2000modelling, ahn2002estimating}; (iii) grey-box models that lie between the two models, e.g.,~\cite{pelkmans2004development}. \revt{White-box models are developed by using domain knowledge in physical and chemical processes relevant to fuel consumption. It is most accurate and allows one to investigate tuning internal physical/chemical processes for higher fuel efficiency. However, the white-box model usually has a large number of parameters and can be expensive to obtain. On the contrary, the black-box model is less accurate than the white-box model, does not allow internal-process tuning, but has fewer parameters to determine.

For our purpose of minimizing fuel consumption and similar to existing eco-routing studies, we only need to model the (external) speed to fuel-consumption relationship, instead of the whole physical model with internal process tuning. As such, black-box models suffice, and they are easier to construct than white-box models. Indeed, the particular black-box model used in our simulation only has four parameters, and the model predicts the fuel consumption pretty well~\cite[Fig. 5]{Deng2016EnergyefficientTT}.}


\revt{To proceed}, we define function $f_{e}(r_{e}):[r_{e}^{l},r_{e}^{u}]\to\mathbb{R}^{+}$
as the \emph{(instantaneous) fuel consumption rate} for the truck traveling over $e$
at speed $r_{e}$~\footnote{This paper considers the setting where the two trucks have identical
fuel consumption rate functions. For example, they are the same model
and carry the same truckload weight. Our analysis and solution can
also be extended to the case with heterogeneous fuel
consumption rate functions. }. \revt{It is well understood that $f_{e}(r_{e})$ can be modeled as a strictly
convex function over the speed range $[r_{e}^{l},r_{e}^{u}]$; see
e.g.,~\cite{Deng2016EnergyefficientTT,Liu2020TITS_energy_truck} with justifications.}
Many existing works \cite{nam2005fuel, Ardenkani2001Traffic, An1993MODELOF, Yue2008MESOSCOPICFC, gao2007modeling, nie2013eco, demir2011comparative, lajunen2014energy} also use polynomial functions to model the \emph{fuel consumption rate} $f_e\left(\cdot\right)$. Following the practice, in this paper, we assume that $f_e\left(\cdot\right)$ is a strictly convex polynomial function.} As a result, according to Jensen’s inequality, driving at a constant speed is most fuel-economic inside a road segment. Thus, without loss of optimality, it suffices to consider trucks traveling at constant speeds over individual road segments~\cite{Deng2016EnergyefficientTT,Liu2020TITS_energy_truck}. \rev{The fuel consumption due to acceleration and deceleration between
consecutive road segments is negligible as compared to that of driving
over individual segments}\footnote{The speed transition spans over only several
hundred feet~\cite{yang2016truck}, while trucks travel on a road segment for several miles or longer~(3.1 miles on average in the US national highway network as shown in Sec.~\ref{sec:simulation}). 
It is reported in~\cite{Liu2020TITS_energy_truck} that the fuel usage due to acceleration and deceleration between consecutive road segments is less than $0.5\%$ of the total fuel consumption.}.
For ease of discussion, we further define the \emph{fuel consumption
function} for a truck to traverse $e$ as 
\begin{equation}
    c_{e}(t_{e})\triangleq t_{e}\cdot f_{e}\left({D_{e}}/{t_{e}}\right),\forall t_{e}\in[t_{e}^{l},t_{e}^{u}],\label{eq:c_e_def}
\end{equation}
where $t_{e}$ is the travel time of the truck and $t_{e}^{l}\triangleq{D_{e}}/{r_{e}^{u}}$
and $t_{e}^{u}\triangleq{D_{e}}/{r_{e}^{l}}$ are the minimum and
maximum travel duration, respectively. We note that $c_{e}(\cdot)$
is strictly convex over $[t_{e}^{l},t_{e}^{u}]$ as its second-order
derivative $c^{\prime\prime}(t_{e})=f^{\prime\prime}\left({D_{e}}/{t_{e}}\right){D_{e}^{2}}/{t_{e}^{3}}$
is positive. The function has taken into account the influence of
road grade, surface type, truck weight, and environment conditions.
We define the \emph{platooning fuel consumption function} when the
two trucks platoon on road segment $e$ as 
\begin{align}
    c_{e}^{p}(t_{e})= & 2(1-\eta)c_{e}(t_{e}),\label{platoonfunction}
\end{align}
where $\eta\in(0,1)$ is the average fuel-saving ratio and it usually
takes value around 0.1~\cite{Larson2013,LARSSON2015258,Larson2016CoordinatedPR}.
\revt{Thus when the two trucks platoon at the same speed, they jointly save a {$2\eta$} fraction of total fuel as compared to driving separately. Note that this model focuses on the total fuel saving; the leading and following trucks can still have different fuel saving performance. Similar to~\cite{ LARSSON2015258, Larson2016CoordinatedPR}, we assume constant fuel-saving rate ratio. More sophisticated fuel models can be derived from Newton’s second law, e.g., those in \cite{Alam2013LCC, Vandehoef2015, Liang2016}, and the fuel-saving rate can be speed-dependent~\cite{Liang2016}. For these model, it can be verified that the corresponding platooning fuel consumption function is still convex and our approach can be applied directly.}

We consider the scenario of two-truck platooning with deadlines. Each
truck is associated with a task, denoted by $(s_{i},d_{i},\alpha_{s}^{i},\beta_{s}^{i}, \alpha_{d}^i, \beta_{d}^{i})$,
$i=1,2$. Here $s_{i}$ and $d_{i}$ are origin and destination for
truck $i$, respectively. The pickup window $[\alpha_{s}^{i},\beta_{s}^{i}]$ defines the
time period for pickup at $s_{i}$, and the delivery window $[\alpha_{d}^i, \beta_{d}^{i}]$
represents the time period for delivery at $d_{i}$. 
We define the traveling time window of truck $i$ ($i=1,2$) as $\left[T_{s}^{i},T_{d}^{i}\right]$,
where $T_{s}^{i}\triangleq\alpha_{s}^{i}$ and $T_{d}^{i}\triangleq\beta_{d}^{i}$
are the earliest departure time and the latest arrival time, respectively.


\begin{figure*}[t]
\centering
\includegraphics[width=\textwidth]{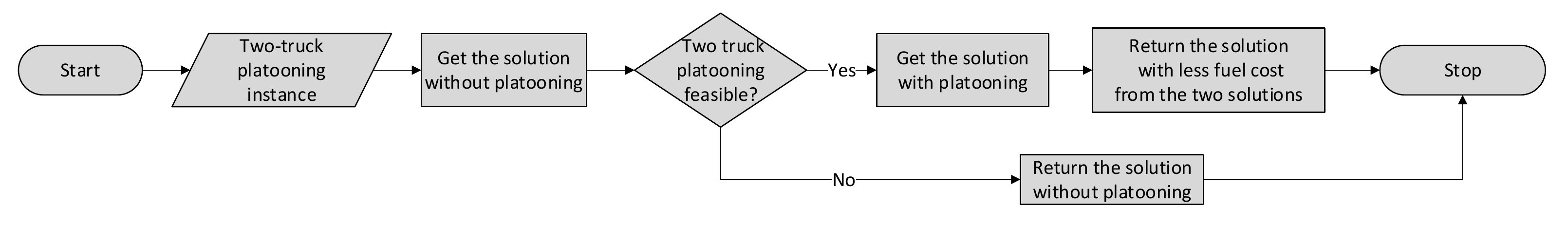}
\caption{The flow chart of our overall solution.}
\label{fig:sysFlow} 
\end{figure*}

\subsection{To Platoon or not to Platoon}
\label{ssec:platooning.or.not}

It may not always be feasible for the two trucks to platoon
without missing deadlines. Even if it is, it may
not necessarily save fuel \revt{when compared to} driving separately.
Given a general two-truck platooning instance, we follow the procedure
shown in Fig.~\ref{fig:sysFlow} to obtain a fuel-minimizing solution.
We first compute the optimized total fuel consumption of two trucks
driving separately, by applying the single-truck algorithm~\cite{Deng2016EnergyefficientTT}
to separately optimize the path and speed planning of individual trucks.
Next, we apply the polynomial-time Algorithm \ref{alg:feasCheck}
proposed in Appendix \ref{subsec:checkFeasi} to check the feasibility
of two-truck platooning under individual deadlines. If infeasible,
the two trucks can only drive separately. We output the separate-driving
solutions computed in the first step. Otherwise, we solve the feasible
two-truck platooning problem with our proposed algorithm to be discussed
later. Afterward, we compare the platooning solution and the separate-driving
solution and output the one with lower fuel consumption.

With the above understanding, in the rest of the paper, we focus on
the feasible two-truck platooning instances.

\subsection{An \rev{Elegant} Problem Structure and Insights}

The conventional formulation of the multi-truck platooning problem
introduces a large number of binary variables, each indicating whether
a group of trucks platoon over a road segment or not. The formulated
problem has a combinatorial structure and is challenging to solve.
Indeed, it is APX-hard\footnote{A problem is APX-hard if it is NP-hard and cannot be approximated within an arbitrary constant factor $(\geq 1)$ in polynomial time unless $\mathbf{P}=\mathbf{NP}$.} and admits no FPTAS\footnote{In particular, the authors in \cite{LARSSON2015258} show that the
multi-truck platooning problem can be reduced to the Steiner tree
problem, which is known to be APX-hard and no FPTAS exists~\cite{BERN1989171,DRAKE200415}.}. Interestingly, the following lemma reveals an \rev{elegant} structure
of the feasible two-truck platooning problem, which allows us to significantly
simplify the problem formulation. \begin{lemma} \label{lem:platoonOnce}
For any feasible two-truck platooning instance, there exists an optimal
solution in which the two trucks platoon only once. \end{lemma} \begin{proof}
The proof is presented in Appendix \ref{lem:platoonOnceproof}. \end{proof}

\begin{figure}[t]
    \centering \includegraphics[width=\columnwidth]{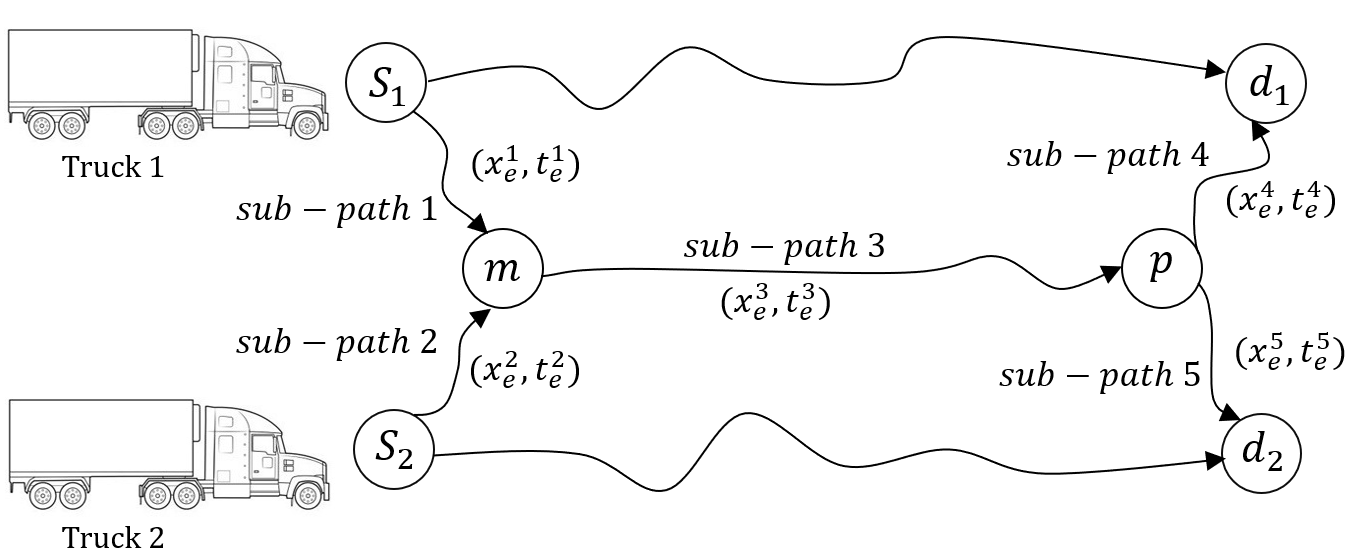}
    \caption{An illustration of our two-truck platooning problem decomposition.}
    \label{fig:two-truck-platooning} 
\end{figure}

Lemma~\ref{lem:platoonOnce} covers the observation in \cite{LARSSON2015258}
as a special case without deadline constraints and speed planning.
It implies that without loss of optimality, we can focus on solutions
in which the two trucks form a platoon only once. This structure significantly
reduces the number of platooning configurations to consider, leading
to a new formulation that admits efficient algorithms. In particular,
applying Lemma~\ref{lem:platoonOnce}, we divide the two-truck platooning
problem into {merging and splitting points selection problem and the following} five sub-problems: 

\vspace{-5pt}

\begin{enumerate}
\item pre-platooning path and speed planning for truck 1;


\item pre-platooning path and speed planning for truck 2;


\item in-platooning path and speed planning for truck 1 and 2;


\item post-platooning path and speed planning for truck 1;


\item post-platooning path and speed planning for truck 2. 
\end{enumerate}
\vspace{-3pt}
 We {denote} the five decomposed paths {by} sub-path 1/2/3/4/5,
{respectively; see Fig.~\ref{fig:two-truck-platooning} for an illustration}.
For  ease of presentation, we denote the set $\left\{ 1,...,k\right\} $
as $\mathbb{N}_{k}$. We introduce binary variables $x_{e}^{i},e\in E,i\in\mathbb{N}_{5}$
for path planning, where 
\[
    x_{e}^{i}=\left\{ \begin{array}{lr}
    1,\quad\text{if edge }e\text{ is selected by the sub-path }i;\\
    0,\quad\text{otherwise.}
    \end{array}\right.
\]
We also introduce non-negative variables $t_{e}^{i}$
to denote the travel time of the (corresponding) trucks over edge $e$
on the sub-path $i$. 
Furthermore, we use binary variables $y_{v}$
and $z_{v}$ to indicate whether the two trucks start and
finish platooning at node $v\in V$, respectively.

\subsection{Problem Formulation}

For simplicity of presentation, we define $\vec{x}\triangleq(x_{e}^{i},i\in\mathbb{N}_5)_{e\in E}$,
$\vec{y}\triangleq(y_{v})_{v\in V}$, $\vec{z}\triangleq(z_{v})_{v\in V}$, $\vec{t}\triangleq(t_{e}^{i},i\in\mathbb{N}_5)_{e\in E}$, and  $\mathcal{T}\triangleq\big\{\vec{t}: t_e^{l}\leq t_e^i\leq t_e^{u}, \forall e\in E, i\in \mathbb{N}_5 \big\}$.
We define $\mathcal{P}$ as the feasible set of $\left(\vec{x},\vec{y},\vec{z}\right)$ as follows: 
\begin{align}
    \mathcal{P}\triangleq\Big\{ & \left(\vec{x},\vec{y},\vec{z}\right)\big|\;\;x_{e}^{i},y_{v},z_{v}\in\{0,1\},\forall e\in E,i\in\mathbb{N}_{5},v\in V,\nonumber \\
     & \sum_{e\in\textbf{Out}(v)}x_{e}^{1}-\sum_{e\in\textbf{In}(v)}x_{e}^{1}=\textbf{1}_{v=s_{1}}-y_{v},\forall v\in V,\nonumber \\
     & \sum_{e\in\textbf{Out}(v)}x_{e}^{2}-\sum_{e\in\textbf{In}(v)}x_{e}^{2}=\textbf{1}_{v=s_{2}}-y_{v},\forall v\in V,\nonumber \\
     & \sum_{e\in\textbf{Out}(v)}x_{e}^{3}-\sum_{e\in\textbf{In}(v)}x_{e}^{3}=y_{v}-z_{v},\forall v\in V,\nonumber \\
     & \sum_{e\in\textbf{Out}(v)}x_{e}^{4}-\sum_{e\in\textbf{In}(v)}x_{e}^{4}=z_{v}-\textbf{1}_{v=d_{1}},\forall v\in V,\nonumber   
     \\
     & \sum_{e\in\textbf{Out}(v)}x_{e}^{5}-\sum_{e\in\textbf{In}(v)}x_{e}^{5}=z_{v}-\textbf{1}_{v=d_{2}},\forall v\in V,\nonumber \\
     & \sum_{v\in V}y_{v}=1,\sum_{v\in V}z_{v}=1\Big\},\label{feasible_region}
\end{align}
where $\textbf{Out}(v)$ and $\textbf{In}(v)$ are the set of outgoing
edges and incoming edges of node $v$, respectively. $\textbf{1}_{\left\{ *\right\} }$
is an indicator variable taking value 1 if the statement $*$ is true
and 0 otherwise. The first five constraints in $\mathcal{P}$ are
the flow conservation requirements for the 5 sub-paths. The constraints
on $y_{v}$ and $z_{v}$ restrict the solution to platoon only once.

We now model the deadline constraints and speed planning. The speed planning is done by adjusting the travel time $t_e^i$ for individual edge $e$. To increase
fuel efficiency, we consider \rev{departure coordination}
where trucks
can strategically wait at the origins to properly align the schedule
to form an efficient platoon later. Consequently, the two trucks start
to platoon from the merging point at time
\begin{align}
\max & \left\{ T_{s}^{1}+\sum_{e\in E}t_{e}^{1}x_{e}^{1},\;\; T_{s}^{2}+\sum_{e\in E}t_{e}^{2}x_{e}^{2}\right\} .
\end{align}
The two \revt{expressions} in the bracket are the arrival times at the merging
point of the two trucks, without \revt{departure coordination}. If one term in
the bracket is smaller than the other, the corresponding truck waits
at the origin or a nearby rest area to synchronize the arrival time
at the merging point. \revt{Without loss of generality, we ignore the cost of strategic waiting at the origin. It can be incorporated into our model by adding a dummy waiting edge at the origin with a designated cost function.} To simplify the formulation later, we define
\begin{subequations}
\label{equ:delta}
\begin{align}
    \delta_{1}(\vec{t},\vec{x})= & T_{s}^{1}+\sum_{e\in E}t_{e}^{1}x_{e}^{1}+\sum_{e\in E}t_{e}^{3}x_{e}^{3}+\sum_{e\in E}t_{e}^{4}x_{e}^{4}-T_{d}^{1}, \label{equ:delta_1}\\
    \delta_{2}(\vec{t},\vec{x})= & T_{s}^{2}+\sum_{e\in E}t_{e}^{2}x_{e}^{2}+\sum_{e\in E}t_{e}^{3}x_{e}^{3}+\sum_{e\in E}t_{e}^{4}x_{e}^{4}-T_{d}^{1}, \label{equ:delta_2}\\
    \delta_{3}(\vec{t},\vec{x})= & T_{s}^{1}+\sum_{e\in E}t_{e}^{1}x_{e}^{1}+\sum_{e\in E}t_{e}^{3}x_{e}^{3}+\sum_{e\in E}t_{e}^{5}x_{e}^{5}-T_{d}^{2}, \label{equ:delta_3}\\
    \delta_{4}(\vec{t},\vec{x})= & T_{s}^{2}+\sum_{e\in E}t_{e}^{2}x_{e}^{2}+\sum_{e\in E}t_{e}^{3}x_{e}^{3}+\sum_{e\in E}t_{e}^{5}x_{e}^{5}-T_{d}^{2}. \label{equ:delta_4}
\end{align}
\end{subequations}
It is straightforward to verify that truck 1 meets its deadline if
and only if $\delta_{1}(\vec{t},\vec{x})\leq0$ and $\delta_{2}(\vec{t},\vec{x})\leq0$.
Similarly, truck 2 meets deadline if and only if $\delta_{3}(\vec{t},\vec{x})\leq0$
and $\delta_{4}(\vec{t},\vec{x})\leq0$.

To this end, we formulate the feasible two-truck platooning problem
as follows: 
\begin{align}
    \min\quad & \sum_{i\in\mathbb{N}_{5}\backslash\left\{ 3\right\}}\sum_{e\in E}c_{e}(t_{e}^{i})x_{e}^{i}+\sum_{e\in E}c_{e}^{p}(t_{e}^{3})x_{e}^{3}\label{equ:objective} 
    \\
    \mathrm{s.t.}\quad & \delta_{j}(\vec{t},\vec{x})\leq0,j\in\mathbb{N}_{4};\label{equ:deadlines}\\
    \mbox{var.}\quad & \left(\vec{x},\vec{y},\vec{z}\right)\in\mathcal{P}, \vec{t}\in\mathcal{T}.\nonumber 
\end{align}
The objective in $\eqref{equ:objective}$ is the total fuel consumption
of the two trucks. 
 $\eqref{equ:deadlines}$ describes their deadline
constraints. $\mathcal{P}$ restricts the variables
so that the two trucks platoon only once. Compared to the conventional
formulation, our formulation in~$\eqref{equ:deadlines}$ leverages
Lem.~\ref{lem:platoonOnce} to simplify the optimization of platooning
configuration, by only considering solutions that platoon only once.
Further, the formulation only involves four deadline constraints.
We show in Sec.~\ref{sec:meth} that this structure significantly
simplifies the design of the dual-based algorithm. 

\subsection{Hardness}
It is not surprising that the two-truck platooning problem is challenging
to solve due to its combinatorial structure and multiple deadline
constraints. Meanwhile, it is easier than the general multi-truck
counterpart. 

\begin{theorem} \label{Thm:NP} The two-truck platooning problem
in $\eqref{equ:objective}$ is NP-hard, but only in the weak sense.
In particular, there is a Fully Polynomial Time Approximation Scheme (FPTAS)
for the two-truck platooning problem that achieves $\left(1+\epsilon\right)$
approximation ratio $\left(\forall \epsilon>0\right)$ with a time complexity
of $\mathcal{O}\left(N^{6}/\epsilon^{4}\right)$. 
\end{theorem}
\begin{proof} 
    See Appendix \ref{Thm:NP_Proof}.
\end{proof}
\revt{Thm. \ref{Thm:NP} and its proof show that the two-truck platooning problem is weakly NP-hard and admits an FPTAS.} 
This is in sharp contrast to the general multi-truck platooning problem (even without speed planning), which is APX-hard and repels any FPTAS.
Consequently, there exists a pseudo-polynomial time algorithm to solve the two-truck platooning problem exactly. 
Further, one can trade optimality loss for running time by using the FPTAS derived in the proof. \revt{It guarantees to return an approximate and feasible solution to the two-truck platooning problem, whose fuel consumption  is within a ratio of $(1+\epsilon)$ to the optimal for any $\epsilon>0$, with a time complexity of $\mathcal{O}\left(N^{6}/\epsilon^{4}\right)$. The FPTAS allows us to solve small-scale two-truck platooning instances efficiently.
Meanwhile, it may still incur excessive running time for solving large-scale instances. 
In the next section, we introduce a dual-subgradient algorithm for solving large-/national- scale instances.} 


\section{An Efficient Dual-Subgradient Algorithm}

\label{sec:meth} In this section, we design an efficient dual-subgradient algorithm
for solving a partially-relaxed dual problem of the two-truck platooning
problem, in which only a subset of the constraints are relaxed. While a fully-relaxed dual problem is always convex, the
partially-relaxed one is not. Indeed, it is still a combinatorial
problem and can be challenging to solve. As a key contribution, we
explore the problem structure to show that it can be solved optimally
in polynomial time by a dual-subgradient algorithm. 
{We derive a posterior bound on the optimality gap when a feasible primal solution is recovered from the dual optimal solution. We resort to the separate-driving solution when no feasible primal solution is recovered.}

\subsection{The Dual of the Two-truck Platooning Problem}
We relax the deadline constraints in $\eqref{equ:deadlines}$ and
obtain the following Lagrangian function 
\begin{align*}
    L\left(\vec{x},\vec{y},\vec{z},\vec{t},\vec{\lambda}\right) &=  \sum_{i\in\mathbb{N}_{5}\backslash\left\{ 3\right\} }\sum_{e\in E}c_{e}\left(t_{e}^{i}\right)x_{e}^{i} +\sum_{e\in E}c_{e}^{p}\left(t_{e}^{3}\right)x_{e}^{3}\\
    & \qquad +\sum_{i\in\mathbb{N}_{4}}\lambda_{i}\delta_{i}\left(\vec{t},\vec{x}\right),
\end{align*}
where $\vec{\lambda}\triangleq\left(\lambda_{1},\lambda_{2},\lambda_{3},\lambda_{4}\right)$
and $\lambda_{i}$ is the dual variable associated with the constraint
$\delta_{i}(\vec{t},\vec{x})\leq0$.  The corresponding dual problem
is given by 
\begin{align}
    \max\limits _{\vec{\lambda}\geq0}D\left(\vec{\lambda}\right)\triangleq & \max\limits _{\vec{\lambda}\geq0}\,\min\limits _{\left(\vec{x},\vec{y},\vec{z}\right)\in\mathcal{P},\vec{t}\in\mathcal{T}}L\left(\vec{x},\vec{y},\vec{z},\vec{t},\vec{\lambda}\right).\label{equ:dualprob}
\end{align}
Our idea is to solve this partially-relaxed dual problem optimally and recover
the corresponding primal solutions. At the first glance, {the inner problem seems challenging as it is still combinatorial
(due to the structure of the feasible set $\mathcal{P}$).} To proceed,
we re-express $D(\vec{\lambda})$ as follows:
\begin{align}
    \negthickspace D\left(\vec{\lambda}\right) & =h\left(\vec{\lambda}\right)+\underbrace{\min_{\underset{\in\mathcal{P}}{\left(\vec{x},\vec{y},\vec{z}\right)}}\sum_{i\in\mathbb{N}_{5}}\sum_{e\in E}x_{e}^{i}\cdot\underbrace{\min_{t_{e}^{i}\in\left[t_{e}^{l},t_{e}^{u}\right]}g_{e}^{i}\left(\vec{\lambda},t_{e}^{i}\right)}_{\mbox{speed optimization}}}_{\mbox{path and platooning optimization}},\label{equ:dualfunction}
\end{align}
 where $h(\vec{\lambda})\triangleq\lambda_{1}\left(T_{s}^{1}-T_{d}^{1}\right)+\lambda_{2}\left(T_{s}^{2}-T_{d}^{1}\right)+\lambda_{3}\left(T_{s}^{1}-T_{d}^{2}\right)$
$+\lambda_{4}\left(T_{s}^{2}-T_{d}^{2}\right)$ is a linear function
and $g_{e}^{i}(\vec{\lambda},t_{e}^{i})$ is a generalized edge cost
function given by, for all $e\in E$,
\[
    g_{e}^{i}\left(\vec{\lambda},t_{e}^{i}\right)=\begin{cases}
    c_{e}\left(t_{e}^{1}\right)+\left(\lambda_{1}+\lambda_{3}\right)t_{e}^{1}, & i=1;\\
    c_{e}\left(t_{e}^{2}\right)+\left(\lambda_{2}+\lambda_{4}\right)t_{e}^{2}, & i=2;\\
    c_{e}^{p}\left(t_{e}^{3}\right)+\left(\sum_{j\in\mathbb{N}_{4}}\lambda_{j}\right) t_{e}^{3}, & i=3;\\
    c_{e}\left(t_{e}^{4}\right)+\left(\lambda_{1}+\lambda_{2}\right)t_{e}^{4}, & i=4;\\
    c_{e}\left(t_{e}^{5}\right)+\left(\lambda_{3}+\lambda_{4}\right)t_{e}^{5}, & i=5.
    \end{cases}
\]
The function $g_{e}^{i}(\vec{\lambda},t_{e}^{i})$, for all $i\in\mathbb{N}_{5}$
and $e\in E$, is convex in $t^i_e$ as the sum of a convex function and
a linear function. As such, the speed optimization module in $\eqref{equ:dualfunction}$, one for each edge, 
can be solved in polynomial time. For all $i\in\mathbb{N}_{5}$
and $e\in E$, define $w_{e}^{i}(\vec{\lambda})\triangleq\min_{t_{e}^{i}\in\left[t_{e}^{l},t_{e}^{u}\right]}g_{e}^{i}(\vec{\lambda},t_{e}^{i})$
and $t^{i,*}_e(\vec{\lambda})$ as the optimal objective value and a
corresponding optimal solution, respectively. We further define $\vec{t}^{\,*}(\vec{\lambda})\triangleq(t_{e}^{i,*}(\vec{\lambda}))_{i\in\mathbb{N}_{5}, e\in E}$.
After obtaining $w_{e}^{i}(\vec{\lambda})$ for all $i\in\mathbb{N}_{5}$ and $e\in E$, the path and platooning optimization module in $\eqref{equ:dualfunction}$ amounts to a combinatorial integer linear programming (ILP) problem: 
\begin{align}
    \min_{\left(\vec{x},\vec{y},\vec{z}\right)\in\mathcal{P}} & \sum_{i\in\mathbb{N}_{5}}\sum_{e\in E}w_{e}^{i}\left(\vec{\lambda}\right)x_{e}^{i}.\label{eq:prob_ppo_ILP}
\end{align}
In general, ILP problems can be difficult to solve.
Generic approaches such as the branch-and-bound method incur high
computational complexity and do not scale well to large problem
instances. As a key contribution, we show that one can relax binary
decision variables $\vec{x},\vec{y},\vec{z}$ into continuous ones
in $\left[0,1\right]$ and solve the resulting linear programming
(LP) problem without loss of optimality.  Recall that $M$ is the number of edges.
\begin{theorem} \label{Thm_ILPGap}
    {There is no integrality gap between the ILP problem in $\eqref{eq:prob_ppo_ILP}$ and the corresponding relaxed LP problem. Furthermore, {the optimal solution to the former can be recovered from that to the latter in $\mathcal{O}(M^2)$ time}.} \label{thm:zeroIntGap} 
\end{theorem}
\begin{proof}
    The idea is to show every optimal solution of the relaxed LP problem
    is a convex combination of integer platooning solutions in $\mathcal{P}$. 
    See Appendix \ref{proof:ILPGap} for the proof. 
\end{proof}
\revt{Thm. \ref{Thm_ILPGap} is an important result. It says we can obtain
the optimal solution of the path and platooning optimization problem
in $\eqref{eq:prob_ppo_ILP}$, by solving its relaxed LP problem,
significantly reducing the computational complexity.} Denote the optimal
solution of the ILP problem in $\eqref{eq:prob_ppo_ILP}$ as $\vec{x}^{\,*}(\vec{\lambda}),\vec{y}^{\,*}(\vec{\lambda}),\vec{z}^{\,*}(\vec{\lambda}).$

To this point, it is straightforward to verify that $\begin{aligned}\delta_{i}(\vec{t}^{\,*}(\vec{\lambda}),\vec{x}^{\,*}(\vec{\lambda}))\end{aligned}
$ is a subgradient of $D(\vec{\lambda})$ with respect to $\lambda_{i}$,
$i\in\mathbb{N}_{4}$. From the discussion above, it is clear that
both $\vec{t}^{\,*}(\vec{\lambda})$ and $\vec{x}^{\,*}(\vec{\lambda})$,
hence $\begin{aligned}\delta_{i}(\vec{t}^{\,*}(\vec{\lambda}),\vec{x}^{\,*}(\vec{\lambda}))\end{aligned}
$, can be computed in polynomial time. This allows us to develop an
efficient \revt{dual-subgradient} algorithm for solving the problem in $\eqref{equ:dualprob}$.

\subsection{The Dual-Subgradient Ascent Algorithm}

The dual-subgradient algorithm is an iterative algorithm for updating
the dual variables towards optimality. Specifically, given $\vec{\lambda}[k]$
as the dual variables after the $k$-th iteration (initially, $\vec{\lambda}[0]$
takes arbitrary non-negative values), the algorithm updates the dual
variables as follows: for all $i\in\mathbb{N}_{4}$,
\begin{align}
\lambda_{i}\left[k+1\right]= & \lambda_{i}\left[k\right]+\phi_{k}\left[\delta_{i}\left(\vec{t}^{\,*}\left(\vec{\lambda}\left[k\right]\right),\vec{x}^{\,*}\left(\vec{\lambda}\left[k\right]\right)\right)\right]_{\lambda_{i}\left[k\right]}^{+}, \label{equ:dualDyna}
\end{align}
where $\vec{t}^{\,*}\left(\vec{\lambda}\left[k\right]\right)$ and
$\vec{x}^{\,*}\left(\vec{\lambda}\left[k\right]\right)$ can be computed
in polynomial time by solving the speed optimization and path and platooning
optimization modules in $\eqref{equ:dualfunction}$, respectively. $\phi_{k}>0$ is the step size and function
$[\mu]_{\nu}^{+}$ is defined as 
\begin{align*}
    [\mu]_{\nu}^{+} & \triangleq\left\{ \begin{array}{ll}
    \mu, & \text{if }\nu>0;\\
    \max\{\mu,0\}, & \text{otherwise.}
    \end{array}\right.
\end{align*}

\begin{algorithm2e}[t]
\SetAlgoVlined
\caption{A Dual subgradient algorithm for optimizing two-truck platooning with deadlines.}
\label{alg:heuristic}
    Set $K$, $k\xleftarrow{}1$, and $\lambda_i[0] \xleftarrow{} 0,\forall i\in \mathbb{N}_4$\\
   {Set $k^u$ as $\mathrm{NULL}$}\\
    \While{ $k \le K$}{
        Compute $\vec{t}^{*}\left(\vec{\lambda}[k]\right)$, $\vec{x}^{*}\left(\vec{\lambda}[k]\right)$, $\vec{y}^*\left(\vec{\lambda}[k]\right)$, $\vec{z}^*\left(\vec{\lambda}[k]\right)$ by solving speed optimization, path and platooning optimization modules in $\eqref{equ:dualfunction}$.  \label{line:LP}\\
        Compute $\delta_i\left(\vec{t}^{*}\left(\vec{\lambda}[k]\right), \vec{x}^{*}\left(\vec{\lambda}[k]\right)\right)$ by $\eqref{equ:delta}$.\\
        \If{$\delta_{i}^{*}(\vec{\lambda}[k])=0,\forall i\in \mathbb{N}_4$}{
            \textbf{return} $\vec{t}^{*}\left(\vec{\lambda}[k]\right)$, $\vec{x}^{*}\left(\vec{\lambda}[k]\right)$, $\vec{y}^*\left(\vec{\lambda}[k]\right)$, $\vec{z}^*\left(\vec{\lambda}[k]\right)$ \label{line:optima_returned}\\
        }
        
        \If{$\delta_{i}^{*}(\vec{\lambda}[k])\leq0,\forall i\in \mathbb{N}_4$}{
             Compute the fuel cost $\rho[k]$ of the solution
             $\vec{t}^{*}\left(\vec{\lambda}[k]\right)$, $\vec{x}^{*}\left(\vec{\lambda}[k]\right)$, $\vec{y}^*\left(\vec{\lambda}[k]\right)$, $\vec{z}^*\left(\vec{\lambda}[k]\right)$ \\
             \If{$k^u$ is $\mathrm{NULL}$ \textbf{or} $\rho[k]\leq \rho[k^u]$}{
                $k^u \xleftarrow{} k$\\
             }
        }
        
        Update $\lambda_i[k+1]$ according to~$\eqref{equ:dualDyna}$, $\forall i\in \mathbb{N}_4$.
        $k \xleftarrow{} k + 1$\\
    }
     \uIf{$k^u$ is not $\mathrm{NULL}$}{
            $k^*\xleftarrow{}k^u$
            }
    \Else{
           $k^*\xleftarrow{}K$ 
    }
    { 
    \textbf{return} $\vec{t}^{*}\left(\vec{\lambda}[k^*]\right)$, $\vec{x}^{*}\left(\vec{\lambda}[k^*]\right)$, $\vec{y}^*\left(\vec{\lambda}[k^*]\right)$, $\vec{z}^*\left(\vec{\lambda}[k^*]\right)$. \label{line:near_optima_returned}
    }
\end{algorithm2e}

The pseudo-code of the dual-subgradient algorithm is presented in Alg. \ref{alg:heuristic}.
Intuitively, the dual variables can be interpreted as the delay price.
The dual function can be interpreted as a generalized
cost involving both the delay and fuel expenses. Higher delay prices lead to shorter truck traveling times, which help to capture deadlines. Meanwhile, higher delay prices also result in higher fuel expense, as seen from the structure of the cost function over individual edges, \emph{i.e.}, $g^i_e(\vec{\lambda},t^i_e)$. 
Alg.~\ref{alg:heuristic} can thus be understood as adjusting the dual variables according to the subgradients, towards the condition of meeting the deadlines and minimizing the fuel consumption. If Alg.~\ref{alg:heuristic} returns at line {\ref{line:optima_returned}} and all the deadline constraints are satisfied\footnote{We note that while the dual-subgradient algorithm converges to the dual optimal (as discussed in the next subsection), the deadline constraints may not be satisfied. Indeed, if all the deadline constraints are satisfied at the dual optimal solution, the duality gap is zero, and the dual optimal value is also primal optimal. In general, this  may not happen, and the duality gap is nonzero.}, then the corresponding primal solutions are also feasible. Otherwise Alg.~\ref{alg:heuristic} returns at line {\ref{line:near_optima_returned}}. {If some deadlines are not met}, we perform primal recovery by fixing the returned path planning and platooning solution $(\vec{x}^*(\vec{\lambda}(K)), \vec{y}^*(\vec{\lambda}(K)), \vec{z}^*(\vec{\lambda}(K)))$ and re-optimize the time to spent on individual edges. From the simulation results reported in Sec.~\ref{sec:simulation}, we observe that such a procedure can recover primary feasible solutions for most instances, with an average optimality gap within 1\% to the optimal. In the rare situation where the primal recovery fails, we resort to the separate-driving solution discussed in Sec.~\ref{sec:model}\ref{ssec:platooning.or.not}.


\subsection{Performance Analysis}\label{ssec:perf.analysis}
\subsubsection{Convergence Rate}

It is known that subgradient algorithms for solving fully-relaxed
convex dual problems converge at a rate of $\mathcal{O}\left(1/\sqrt{K}\right)$ \cite{boyd2008subgradient}.
A similar understanding holds for ours for
solving the partially-relaxed dual problem in $\eqref{equ:dualprob}$. 
\begin{theorem} \label{thm:dualConvRate} 
    Let $D^{*}$ be the optimal dual value and
    $\bar{D}_{K}$ be the maximum dual value observed until the K-th iteration
    by running Alg.~\ref{alg:heuristic}. Then with step sizes $\phi_{i}=1/\sqrt{{K}}$, $i\in\mathbb{N}_{K}$,
    there exists a constant $\xi>0$ such that 
    $
    D^{*}-\bar{D}_{k}\leq\xi/\sqrt{K}.
    $
\end{theorem}
\begin{proof}
    See Appendix~\ref{Proof:dualConvRate}. 
\end{proof}
{The constant $\xi$ depends on the network structure and the deadline constraints. 
We refer interested readers to the proof for more detailed discussions.
Thm.~\ref{thm:dualConvRate} shows that we can achieve $\mathcal{O}\left(1/\sqrt{K}\right)$
convergence rate by using a constant step size. In practice, one may achieve faster convergence by using adaptive step sizes~\cite{BAZARAA1981choice}.}


\subsubsection{Complexity}
The main computational complexity involved in Alg.~\ref{alg:heuristic}
lies in solving the LP in line \ref{line:LP}, \revt{which is $\mathcal{O}((N + M)^{2.5})$ by the algorithm in~\cite{Vaidya1989}. Note that $N$ is the number of nodes, $M$ is the number of edges in the graph, and $\mathcal{O}(N + M)$ is the number of variables in the LP.}
Since the dual-subgradient algorithm converges
at a rate of $\mathcal{O}(1/\sqrt{K})$, the overall time complexity
for Alg.~\ref{alg:heuristic} to generate a dual solution within accuracy $\epsilon>0$
to the optimal is $\mathcal{O}((N+M)^{2.5}/\epsilon^{2})$.

\subsubsection{Optimality Gap}
{We provide a posterior bound on the optimality loss when Alg.~\ref{alg:heuristic} returns a feasible primal solution. Recall that $\vec{\lambda}[k^*]$ consists of the dual variables upon Alg.~\ref{alg:heuristic} terminating after $K$ iterations and $\delta_i(\vec{t}^*(\vec{\lambda}[k^*], \vec{x}^*(\vec{\lambda}[k^*]))$ is the corresponding dual subgradient.}

\begin{theorem} \label{Thm_Optgap} 
    {If Alg.~\ref{alg:heuristic} returns a feasible solution $\bar{P}$, the optimality gap between the cost of $\bar{P}$, denoted as $c(\bar{P})$, and the optimal cost, denoted as $OPT$, is bounded as follows:}
    \begin{equation} 
        c\left(\bar{P}\right)-OPT\leq\begin{cases}
        0, & \text{if }\delta_i(\cdot)=0,\forall i\in \mathbb{N}_4 \mbox{ and } 
         \text{Alg. \ref{alg:heuristic}}\\
         & \text{ returns in line \ref{line:optima_returned}};\\
        B, & \mbox{otherwise}.
        \end{cases}        
    \end{equation}
    where {$B=-\sum_{i=1}^{4}\lambda_{i}[k^*] \delta_{i}(\vec{t}^{\,*}(\vec{\lambda}[k^*]),\vec{x}^{\,*}(\vec{\lambda}[k^*]))\geq0$.}
\end{theorem}
\begin{proof}
    The proof is presented in Appendix \ref{Thm:PosteriorBound}. 
\end{proof}
According to Thm. \ref{Thm_Optgap}, if $\delta_i(\cdot)=0,\forall i\in \mathbb{N}_4$, i.e., the deadlines for individual trucks are all met exactly, and thus Alg. \ref{alg:heuristic} returns
in the line \ref{line:optima_returned}, the obtained dual solution is
optimal and the corresponding primal solution is feasible and optimal. Otherwise, Alg. \ref{alg:heuristic} returns in the line
\ref{line:near_optima_returned} and we can compute a posterior bound if the corresponding primal solution is feasible. 
From the simulation results reported in Sec.~\ref{sec:simulation}, we observe that for the cases that our algorithm returns feasible primal solution, the computed posterior bounds are within $0.3\%$ to the optimal, suggesting strong empirical performance of the proposed algorithm. Meanwhile, it is clear that the bound holds for each iteration and thus can be used to terminate the Alg.~\ref{alg:heuristic} earlier \revt{upon reaching a target accuracy threshold.} Specifically, one can calculate a posterior bound after each iteration and terminate Alg. \ref{alg:heuristic} if the bound is already lower than the target threshold.

\section{Performance Evaluation}
\label{sec:simulation}

We use real-world traces to evaluate the performance of our solutions. \revt{We implement our algorithm in python and carry out numerical experiments on a server cluster with 42 pieces of 2.4 GHz -- 3.4 GHz Intel/AMD processors, each equipped with 16GB -- 96GB memory.} We represent an instance of our platooning problem by a tuple $(s_1, d_1, T_s^1, T_d^1, s_2, d_2, T_s^1, T_d^2)$. We follow the flow chart in Fig.~\ref{fig:sysFlow} to obtain the solution for each instance.


\emph{Transportation network and heavy-duty truck}.
We construct the US national highway network using the data from the Clinched Highway Mapping Project~\cite{chmproject}. {We} focus on the eastern US section with $38,213$ nodes and $82,781$ road segments. {The average length of road segments is $3.1$ miles.} Our simulated trucks are both class-8 heavy-duty {trucks} Kenworth T800, {each} with a 36-ton full load~\cite{kent800}. 

\emph{Fuel consumption rate function}. First, we obtain the grade of each road segment by querying the elevations provided by the Elevation Point Query Service~\cite{usgs}. We then use the \textsf{ADVISOR} simulator~\cite{markel2002advisor} to collect fuel consumption rate data {for each road segment}. Finally, we use the curve fitting toolbox in \textsf{MATLAB} to fit the fuel consumption rate function of speed by a 3-order polynomial function. The \revt{same polynomial fuel consumption model} was also used in other literature of energy-efficient timely truck transportation; see, \revt{e.g.,~\cite{Deng2016EnergyefficientTT, xu2019ride, gao2007modeling, nie2013eco, demir2011comparative, lajunen2014energy, Liu2020TITS_energy_truck}.}

\emph{Platooning fuel consumption}.
For the platooning problem, {existing }studies on path-only optimization usually assume that the fuel-saving ratio is constant \cite{Larson2016CoordinatedPR, LUO2018213}. Studies on speed-only planning often assume the aerodynamic coefficient in the fuel consumption model is reduced at a constant rate \cite{Vandehoef2015, Vandehoef201509, Vandehoef2018}. In our simulation, similar to \cite{Larson2016CoordinatedPR, LUO2018213}, \revt{we assume a speed-independent \textit{average} fuel saving rate for the two trucks, denoted as $\eta$ in~$\eqref{platoonfunction}$. This is based on the understanding that (i) the fuel saving rates for the leading and trailing vehicle can be different and we are working with their average and (ii) existing studies, e.g.,~\cite{mcauliffe2017fuel}, show that the impact of speed on platooning fuel saving rate is minor. Further, our algorithm is applicable without any change as long as the platooning fuel consumption function is convex over the speed on each edge.} We use $\eta = 0.1$ in simulations. 

\emph{Origin-destination pairs}. {There are few platooning traces available in the public domain. We thus generate synthetic platooning instances from real-world statistics as follows}. We first obtain the real-world freight flow statistics in the US~{,}from Freight Analysis Framework~(FAF)~\cite{FAF4}. 
{For each origin or destination of the recorded freight flow in FAF, }
we find the closest highway node and manually set this node as the representative truck origin or destination in our simulation.
We then randomly sample the first truck origin-destination $(s_1, d_1)$ by the freight weight from the origin to the destination. We then calculate the shortest path's distance $l_1$ from $s_1$ to $d_1$. We set a radius lower bound ratio $\gamma_l$ and a radius upper bound {ratio} $\gamma_u$. We then set the radius lower bound $r_l=l_1\gamma_l$ and the radius upper bound $r_u=l_1\gamma_u$. We then sample $s_2$~($d_2$ resp.) from the nodes with distance between $r_l$ to $r_u$ from the $s_1$~($d_1$ resp.), {uniformly at random}. \revt{We set $\gamma_l=0.1$ and $\gamma_u=0.5$. In our simulation, the average distances from $s_1$ to $s_2$ and from $d_1$ to $d_2$ are  195 miles and 188 miles, respectively.} 

\begin{table*}
   \caption{Fuel saving performance of various platooning solutions, under different \emph{shortest path overlapping ratios}. } 
    \label{tab:ORImpactPlatooning}
    
    \centering
    \begin{tabularx}{\textwidth}{|X|c|c|c|c|c|}
     
    \hline
    \textbf{Shortest path overlapping ratio} & \textbf{(0.0, 0.2)} & \textbf{(0.2, 0.4)} & \textbf{(0.4, 0.6)} & \textbf{(0.6, 0.8)}&\textbf{(0.8, 1.0)} \\
    \hline  
    P-Platooning platooning ratio (\%) & 45.3 & 56.6 &  63.8 & 76.0 & 85.2 \\
    \hline 
    S-Platooning platooning ratio (\%) & 3.3 & 29.6 & 50.2 & 69.7 & 84.7 \\
    \hline
    Our Solution platooning ratio (\%)  & 43.9 & 56.8 & 63.4 & 75.4 & 85.2 \\
    \hline 
    Fuel Saving compared to P-Platooning (\%)& 22.4 & 22.8 & 28.2 & 25.3&20.2 \\
    \hline
    Fuel Saving compared to S-Platooning (\%) & 4.3& 2.6 &1.2 & 0.8& 0.2\\    
    \hline
  {Upper bound of our solution's optimality gap} (\%) &0.2 &0.3 & 0.5&0.5 & 0.5\\  
    \hline 
    \end{tabularx}

\end{table*}

\subsection{Two Platooning Baselines}
We compare the performance of our dual-subgradient algorithm to two baseline methods.   
\begin{itemize}
    \item \emph{Path planning with platooning (P-Platooning).}
    The first one is based on path planning with the speed of trucks on each edge fixed \cite{Larson2016CoordinatedPR, LUO2018213}. The platooning optimization can be formulated as an integer linear programming problem (ILP) and solved {by the state-of-the-art solver Gurobi}. Specific to our setting, we implement the {P-Platooning} baseline by setting the speed as the fastest speed and then optimizing the {driving paths with platooning consideration}. 
    
    \item \emph{Speed planning with platooning (S-Platooning).}
    The second one is based on speed planning. Similar to  \cite{Vandehoef2015, Vandehoef201509, Vandehoef2018}, we assume the two trucks travel on their shortest paths. We consider the two trucks may platoon over the overlapping road segments of the two shortest paths. We then optimize speed planning with platooning in consideration subject to deadline constraints, by fixing paths in our formulation in~$\eqref{equ:objective}$ and using Gurobi to solve the revised problem. 
\end{itemize}
We use the largest dual function value obtained from our dual subgradient algorithm as a non-trivial lower bound to estimate the minimum fuel consumption for evaluating optimality loss.

\subsection{Performance Comparison to Platooning Baselines} 
Intuitively, by jointly exploiting the design space of path planning, speed planning, {and platooning configuration}, we can achieve better performance than only {optimizing the path planning or the speed planning.} 
{We} sample the second origin-destination pair within $0.1-0.5$ times the shortest path distance of the first distance pair. {We define the platooning ratio as the average fraction of platooning in driving distance, and the \emph{Shortest Path Overlapping Ratio{ (S.P.O.R.)}} as the ratio between the length of overlapping road segments in shortest paths and the {total} length of two {trucks'} shortest paths. We then group the results according to the range of \emph{{the overlapping ratio}.}}  

We also introduce a new evaluating metric \emph{platooning fuel saving achieving ratio $P$} defined as {the ratio between $\mathrm{OPT(Path, Speed)}-\textrm{Fuel(Platoon  Solution)}$ and $\mathrm{OPT(Path, Speed)}-\mathrm{OPT(Path, Speed, Platoon)}$}.
{$\mathrm{OPT(Path, Speed)}$ is the optimal fuel cost of path planning and speed planning without platooning, {\textrm{Fuel(Platoon Solution)} is the fuel cost of the platoon solution to be evaluated} and $\mathrm{OPT(Path, Speed, Platoon)}$ is the optimal fuel cost of path planning and speed planning with platooning.} 
{We obtain estimates of the optimal values by using {the largest dual value obtained from the dual subgradient algorithm}.} 
\footnote{{{Specifically, we obtain the approximated $\mathrm{OPT(Path, Speed})$ by using the dual based method in~\cite{Deng2016EnergyefficientTT}. We use the largest dual value obtained in our solution, which is within $0.8\%$ to $\mathrm{OPT(Path, Speed, Platoon)}$, to approximate $\mathrm{OPT(Path, Speed, Platoon)}$.}}}
\begin{figure}
    \centering
    \includegraphics[width=\columnwidth]{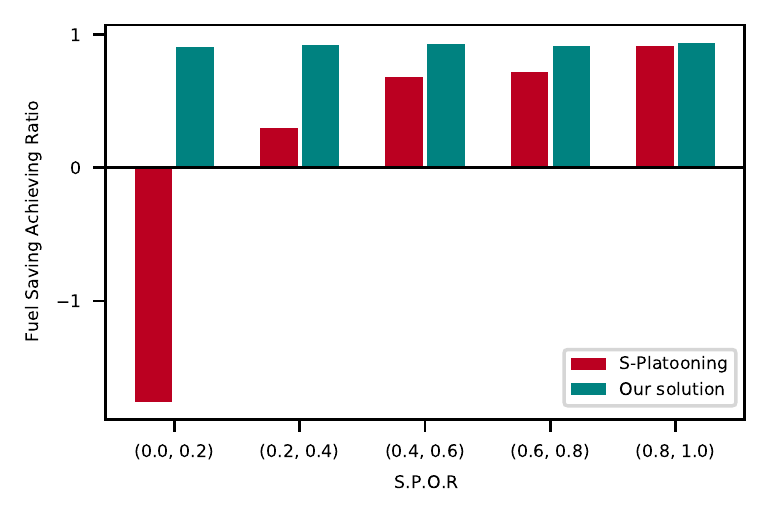}
    \caption{\emph{Platooning fuel saving achieving ratio} with different \emph{Shortest Path Overlapping Ratio}. The achieving ratio of \revt{S-Platooning} decreases to negative, as the overlapping ratio decreases. In contrast, our solution attains an achieving ratio close to one consistently.} 
    \label{fig:diff_overlap_achieve_rate}
\end{figure}

We simulate $1,495$ instances in total. Simulation results in Tab.~\ref{tab:ORImpactPlatooning} show that on average, our solution saves $24\%$~(resp. $3\%$) fuel compared to the P-Platooning~(resp. S-Platooning) baselines. We observe that, empirically, our solution's fuel {consumption} is within less than $0.8\%$ {to} the lower bound of {minimum}. The optimality loss {of} our solution is {thus} minor.

Furthermore, Fig.~\ref{fig:diff_overlap_achieve_rate} shows that when \emph{Shortest Path Overlapping Ratio} is small, {the achieving ratio of S-Platooning is minor or even negative.} {In} contrast, our solution always achieve a high ratio close to {one}. This result shows that our solution achieves almost the maximum fuel saving offered by platooning.

\subsection{The Benefits of \rev{Departure Coordination}}

In this experiment, we compare the solution with departure coordination to that without departure coordination. Similar to the opportunistic driving idea proposed in~\cite{xu2019ride}, allowing trucks {to wait at the origins can create more favorable platooning opportunity at a later time and thus benefit fuel saving.} This {improves the fuel economy} as long as the two trucks can still catch their deadlines. {To obtain the solution without departure coordination, {we require} the two trucks to arrive at the merging point exactly the same time, \emph{i.e.}, $T_s^1 + \sum_{e \in E} t_e^1x_e^1 = T_s^2 + \sum_{e \in E} t_e^2x_e^2$, and the other constraints are the same as in $\eqref{equ:deadlines}$. As for the solution with opportunistic driving, we do not require arrival at the merging point exactly the same time.} We use similar dual-based procedures to optimize the path and speed {planning for the case without departure coordination}.


Simulation results show that, {with departure coordination, we save $5\%$
more fuel as compared to the case without departure coordination. We also
observe that about $25\%$ of the instances that are feasible with departure coordination
become infeasible without departure coordination.
{Fig.~\ref{fig:DFImpactRatioOPandWaitTime} shows that the normalized fuel cost with departure coordination is significantly
lower and decreases faster than that without departure coordination.} Meanwhile, we define the waiting time ratio as the ratio of waiting time to the total trip time for the truck that {performs departure coordination} to measure the extent of departure coordination. As shown in Fig.~\ref{fig:DFImpactRatioOPandWaitTime}, the waiting time ratio increases as the delay factor increases. This set of results highlight that as the deadline increases, the time schedule is more flexible for the trucks to form a favorable platoon for fuel saving, and departure coordination is effective in capitalizing the fuel-saving potential.}

\begin{figure}[tbp]
\centering
\includegraphics[width=\columnwidth]{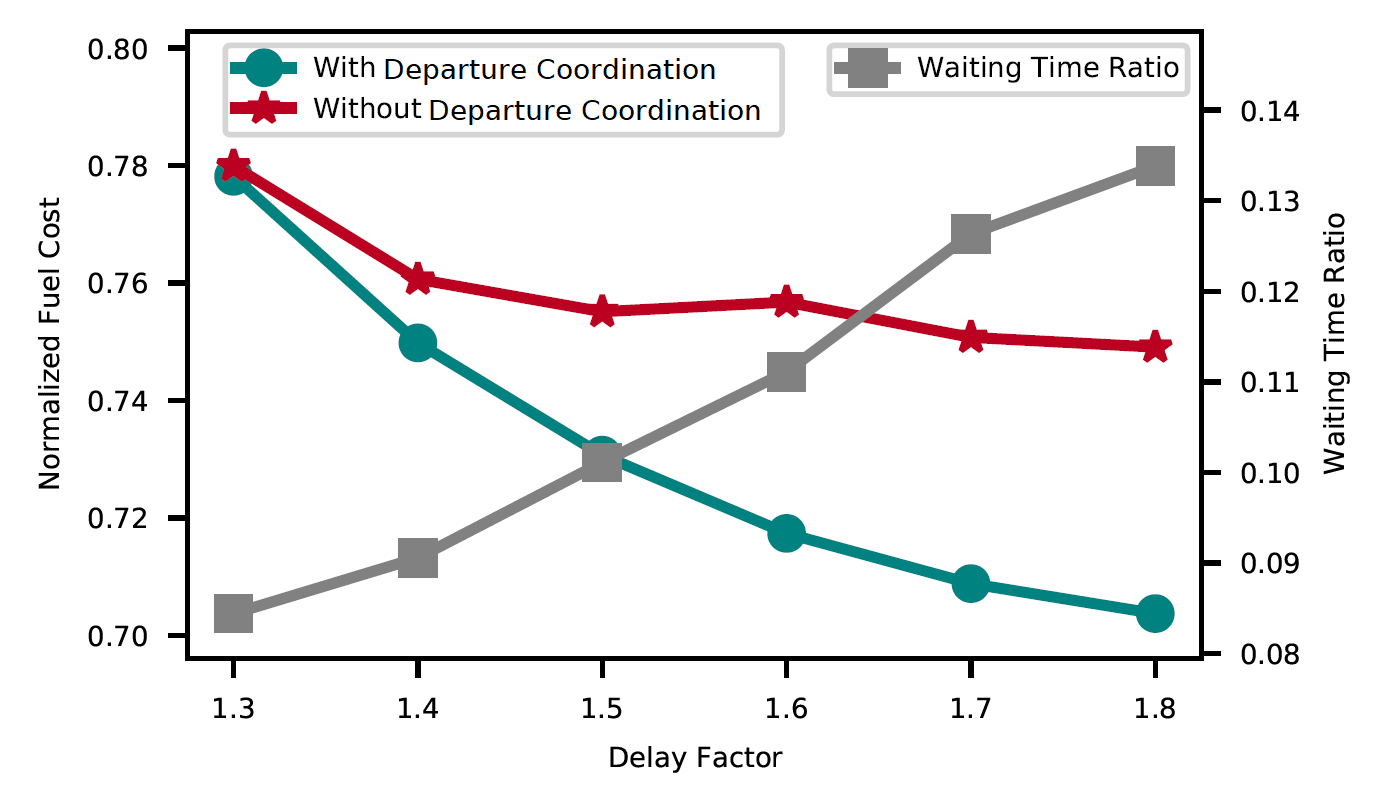}
\vspace{-1ex}
\caption{ The average fuel cost {normalized by the fuel cost of fastest path solution} decreases as the delay factor increases. The average waiting time ratio for the waiting truck increases as delay factor increases.  
}
\label{fig:DFImpactRatioOPandWaitTime}
\end{figure}




\subsection{Impact of Deadline}
We {now} investigate the impact of deadlines. We define delay factor $\beta$ to be the ratio between $\textrm{Deadline}$ and $\textrm{Fastest Path's Time Cost}$. We fix the radius ratio lower bound to be $\gamma_l=0.1$ and the radius ratio upper bound to be $\gamma_u=0.5$. Fig.~\ref{fig:DFImpact} gives the fuel-saving performance of different algorithms as the delay factor varies. It can be seen that our solution {saves a} significant amount of fuel {as compared to} other solutions. It can also be seen that as the delay factor {increases}, the fuel-saving contributed by speed planning becomes larger. {This matches our intuition that a larger delay factor allows a bigger design space and thus larger fuel-saving potential for speed planning.} 
From Tab.~\ref{tab:DFImpactPlatooningTimeRatio}, we also observe that as the delay factor increases, fuel-saving increases significantly but the ratio of platooning time to total driving time {does not increase much}. It {implies} that the increased fuel saving is mainly contributed by speed planning.

\begin{figure}[tbp]
\centering
    \includegraphics[width=\columnwidth]{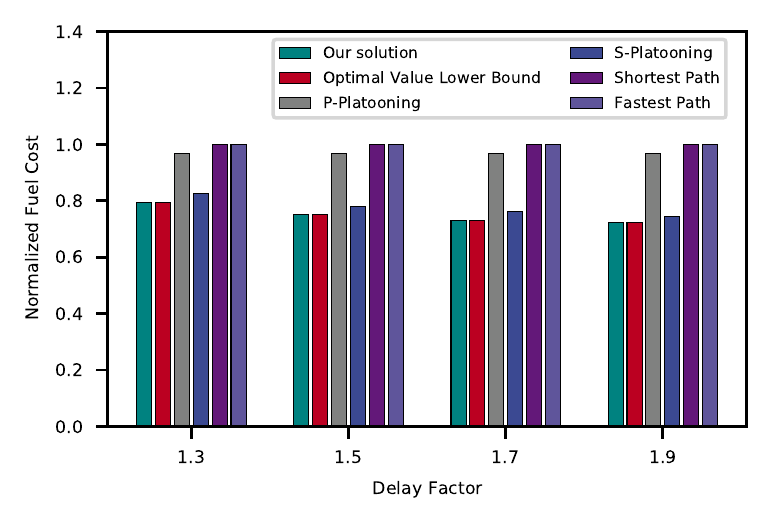} 
\vspace{-1ex}
    \caption{The average fuel cost {normalized by the fuel cost of fastest path solution} of different algorithms with different delay factor. It is again verified that our solution is very close to optimal solution. 
}
\label{fig:DFImpact}
\end{figure}


\begin{table}[t]
    \caption{Comparison of different solutions' fuel saving percentage with different delay factors. The row "Platooning" denotes the Platooning Time Percentage of our solution. It can be seen that there is a significant fraction of platooning time (more than $20\%$), which shows that platooning indeed contributes to fuel saving.}
    \label{tab:DFImpactPlatooningTimeRatio}
    
    \centering
    \begin{tabular}{|c|c|c|c|c|}
    \hline
     \textbf{Delay Factor} & \textbf{1.3}  & \textbf{1.5}   & \textbf{1.7}   & \textbf{1.9}   \\
     \hline
     P-Platooning~(\%) & 3.4 & 3.4 & 3.4 & 3.4\\
     \hline
        S-Platooning~(\%) & 17.9 & 22.7 & 25.0 & 25.9\\
     \hline
     Our solution~(\%) & 20.5 & 25.0 & 27.2 & 28.0\\
     \hline
     Platooning~(\%) & 35.5 & 37.3 & 38.1 & 37.8 \\
    \hline 
    \end{tabular}

\end{table}

    


{We also compare our solution's performance to {S-Platooning} in terms of achieving platooning fuel saving potential, we also compare different solutions' approximated \emph{platooning fuel saving achieving ratio} as shown in Fig.~\ref{fig:diff_df_achieve_rate}. {We observe that S-Platooning's achieving ratio is rather low, sometimes even negative, indicating that the solution is worse than the baseline without platooning. In contrast, our solution always achieve fuel saving ratio close to 1.}} 
\begin{figure}[t]
    \centering
    \includegraphics[width=\columnwidth]{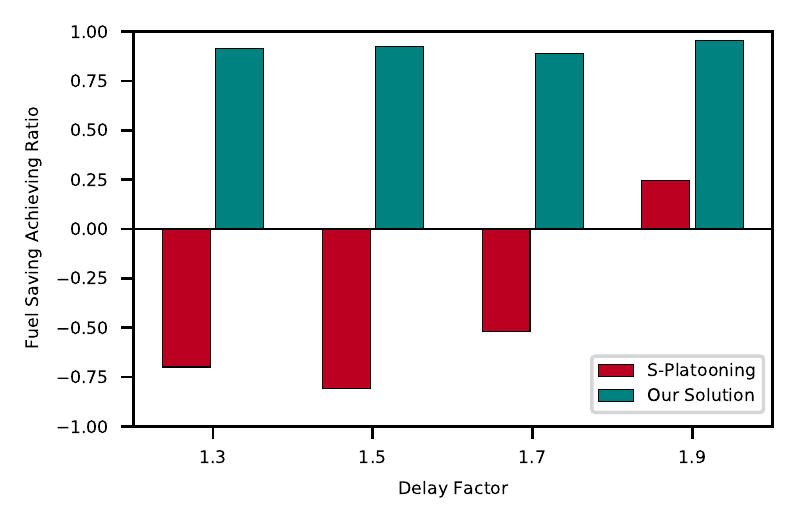}
    \caption{\emph{Platooning fuel saving achieving ratio} under different delay factors. Our solution consistently achieves a ratio close to one, while S-platooning perform badly, even worse than the baseline without platooning.}
    \label{fig:diff_df_achieve_rate}
\end{figure}

\revt{
\section{Discussion}
\subsection{Computation time}
We run all the instances (about 1500 in total) on a server cluster with 42 pieces of 2.4 GHz -- 3.4 GHz Intel/AMD processors, each equipped with 16GB -- 96 GB memory. The total solving time is about $6$ hours, less than 15 seconds for one instance even by the most conservative estimate. On a personal computer with \textsf{Intel(R) Core(TM) i7-10610U} processor and 16 GB RAM, it takes 10-20 minutes to solve one instance.  

In Sec.~\ref{ssec:perf.analysis}, we analyze the time complexity of our proposed algorithm, which is $O((N+M)^{2.5}/\epsilon^2)$ where $N$ is the number of nodes and $M$ is the number of edges in the graph. In simulation, we observe that a large fraction of the computation time is used in obtaining the dual value by solving the relaxed problem in line~\ref{line:LP} of Alg.~\ref{alg:heuristic}, 
which contains two parts. The first is solving a convex optimization problem for each edge. The second is a linear programming with variable dimension proportional to the sum of number of nodes and number of edges. Thus the computation time depends heavily on the size of the graph. As such, one way to reduce the computation time is to prune the nodes and edges that are far away from the origins and destinations of the trucks, which are unlikely to appear in the optimal or near-optimal solution.

\subsection{Perfect and imperfect knowledge of traffic conditions}

As the first step to optimize fuel consumption for two-truck platooning under individual deadline constraints, we solve the fuel minimization problem under the setting where traffic conditions are static and known. The obtained solution is feasible and close to optimal when the static traffic condition can be accurately estimated. They can also serve as benchmarks for existing or future two-truck platooning solutions with imperfect static traffic information. 

Meanwhile, the real-world traffic conditions can also be time-varying. The current formulation does not consider such variable traffic conditions. A viable workaround is to recompute the solution periodically, e.g., once every hour, using the latest traffic information. It is also conceivable to employ the phase-based traffic model in~\cite{xu2019ride} to incorporate variable traffic conditions into our formulation. The obtained results can be benchmarks for evaluating two-truck platooning solutions with imperfect time-varying traffic information.
}
\section{Conclusion and Future Work}
\label{sec:conclusion}
\rev{This paper studies the two-truck platooning problem to minimize the total fuel consumption while meeting individual deadlines by jointly optimizing path planning, speed planning, and platooning configuration.}
We show that any feasible two-truck platooning problem has an optimal solution in which the two trucks platoon only once. With this understanding, we develop a new formulation and show that the two-truck platooning problem is only weakly NP-hard and admits an FPTAS. This \rev{contrasts with} the general multi-truck platooning problem, which is known to be APX-hard and thus repels any FPTAS. We further design a dual-subgradient algorithm for solving large-/national- scale problem instances. \rev{It is an iterative algorithm that always converges. Each iteration involves solving a \revt{non-trivial} integer linear problem optimally, which we prove incurs only polynomial time complexity.} We characterize a sufficient condition under which the algorithm generates an optimal solution.
We characterize a posterior bound on the optimality gap when the condition is not met. We use the real-world traces over the US national highway system to demonstrate the performance of \rev{our algorithm}. As compared to the path-/speed- only baselines adapted from state-of-the-art schemes, our algorithm achieves up to $24\%$ fuel saving \rev{(from path, speed planning and platooning)}, with an average optimality gap no larger  than $0.8\%$. The results demonstrate the effectiveness of our approach. \revt{It is an interesting direction to investigate how the approach in this paper can be extended to the general k-truck platooning settings, by for example dividing the k trucks into multiple groups of two and separately applying the approach to each group for fuel consumption minimization. }

\printbibliography

\section{Appendix}
\subsection{Platooning Example}
\label{Platooning_example}
Here, we introduce a simple example to illustrate the benefits of considering all the design spaces in this paper. All the parameters of this simple example are presented in Fig. \ref{fig:instance}.  The comparison is summarized in Tab. \ref{tab:2truck_optimal}.
\begin{figure}[h]
    \centering
    \includegraphics[width = \columnwidth]{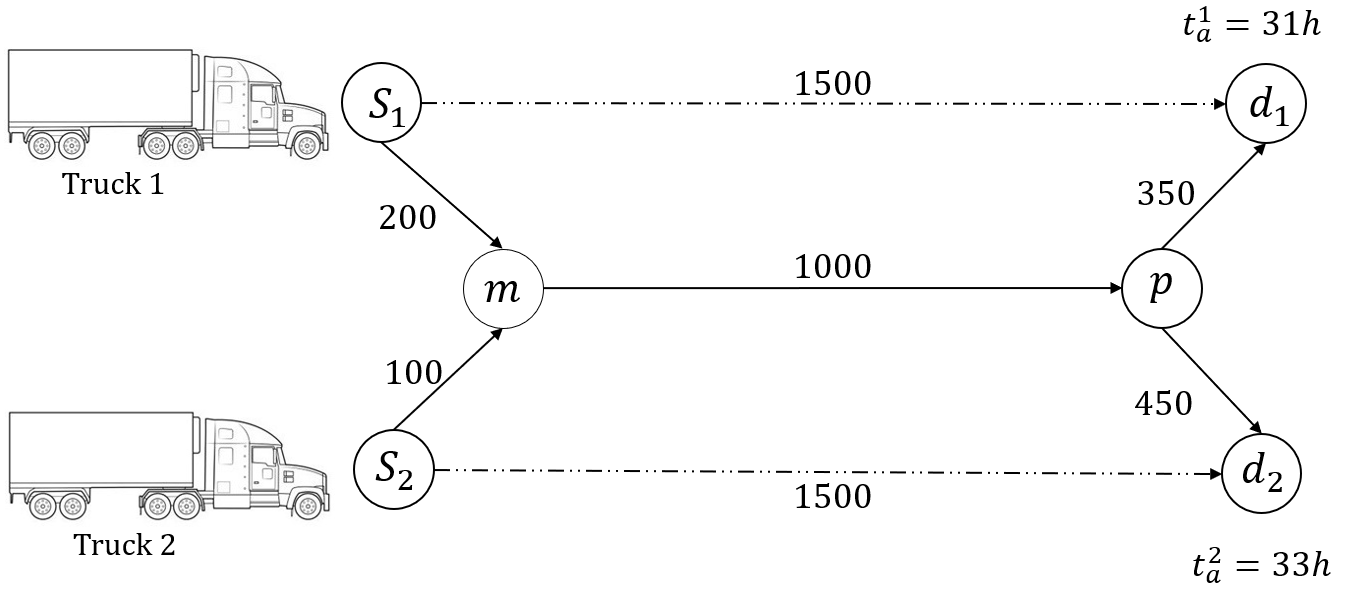}
    \caption{A two-truck platooning instance with \emph{instant fuel consumption rate defined as fuel consumption per unit time function} $f_e(v) = 4\cdot 10^{-4}v^2 - 6\cdot 10^{-3}v + 1$, where $v$ is the driving speed on road segment $e$. \emph{Then total fuel cost function} for each edge is $c_e(t) = 4\cdot 10^{-4}(D_e/t)^2 - 6\cdot 10^{-3}(D_e/t) + 1$, where $D_e$ is the length of the road segment $e$ and $t$ is the driving time spent on $e$. The optimal speed on each edge is 50 mile per hour. The solid lines indicate the optimal platooning solution.}
    \label{fig:instance}
\end{figure}
 \begin{table*}[h]
    \centering
    \begin{tabular}{|c|c|c|c|}
    \hline
      &  Shortest Path & Non-Departure Coordination & Departure Coordination \\
    \hline
        $\mathcal{P}_1$   & $s_1 \rightarrow d_1$& $s_1 \rightarrow m \rightarrow p \rightarrow d_1$ & $s_1 \rightarrow m \rightarrow p \rightarrow d_1$\\
    \hline
        $\mathcal{P}_2$ & $s_2 \rightarrow d_2$& $s_2 \rightarrow m \rightarrow p \rightarrow d_2$ & $s_2 \rightarrow m \rightarrow p \rightarrow d_2$\\
    \hline
        $\mathcal{V}_1$ & 50 & 63.25 $\rightarrow$ 50 $\rightarrow$ 50 & 50 $\rightarrow$ 50 $\rightarrow$ 50 \\
    \hline
        $\mathcal{V}_2$ & 50 & 31.62 $\rightarrow$ 50 $\rightarrow$ 50 & 50 $\rightarrow$ 50 $\rightarrow$ 50 \\
    \hline
        $\mathcal{T}_1$ & 30 &  $3.16 \rightarrow 20 \rightarrow 7 $&  $4 \rightarrow 20 \rightarrow 7$\\
    \hline
        $\mathcal{T}_2$ & 30 & $3.16 \rightarrow 20 \rightarrow 9$ & wait 2 hours $\rightarrow 2 \rightarrow 20 \rightarrow 9$ \\
    \hline
        Average fuel cost & 51.00
 & 49.62 & 49.30\\
    \hline
        Total fuel saving rate & N.A. & 2.7 \% &3.3\%\\
    \hline
    \end{tabular}
    \caption{Optimal solutions for the two-truck platooning instance under different scenarios. $\mathcal{P}_1$ (respectively $\mathcal{P}_2$) refers to optimal paths for truck 1 (respectively truck 2); $\mathcal{V}_1$ (respectively $\mathcal{V}_2$) refers to optimal speeds on each edge for truck 1 (respectively truck 2); $\mathcal{T}_1$ (respectively $\mathcal{T}_2$) refers to optimal travel time on each edge for truck 1 (respectively truck 2).  Compared to the shortest path solution, the total fuel-saving rate for each truck is 3.3\% when the platooning-fuel-saving rate in a platoon is 10\%. {However,} the total fuel-saving rate for each truck is only approximately {2.7\%} when departure coordination is not allowed.}
    \label{tab:2truck_optimal}
\end{table*}

The average daily long-haul traffic intensity on the US national highway in~\ref{fig:22region} shows that most traffic concentrate on some main highways. Platooning is likely to happen between these vehicles, therefore saves fuel. 
\begin{figure}[htbp!]
    \begin{subfigure}{\columnwidth}
    \centering
    \includegraphics[width=\columnwidth]{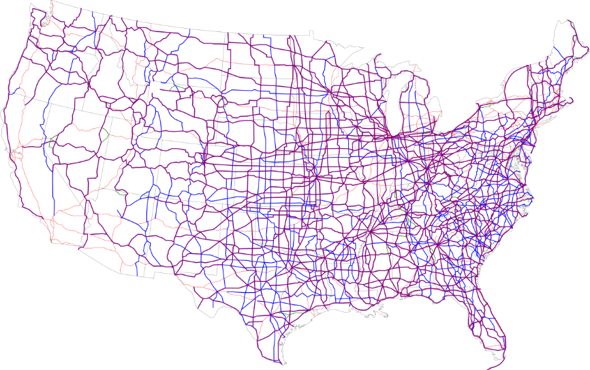}
    \vspace{-1ex}
    \subcaption{$\;$}
    \label{fig:RestAreas}
    \end{subfigure}
    \begin{subfigure}{\columnwidth}
    \centering
    \includegraphics[width=\columnwidth]{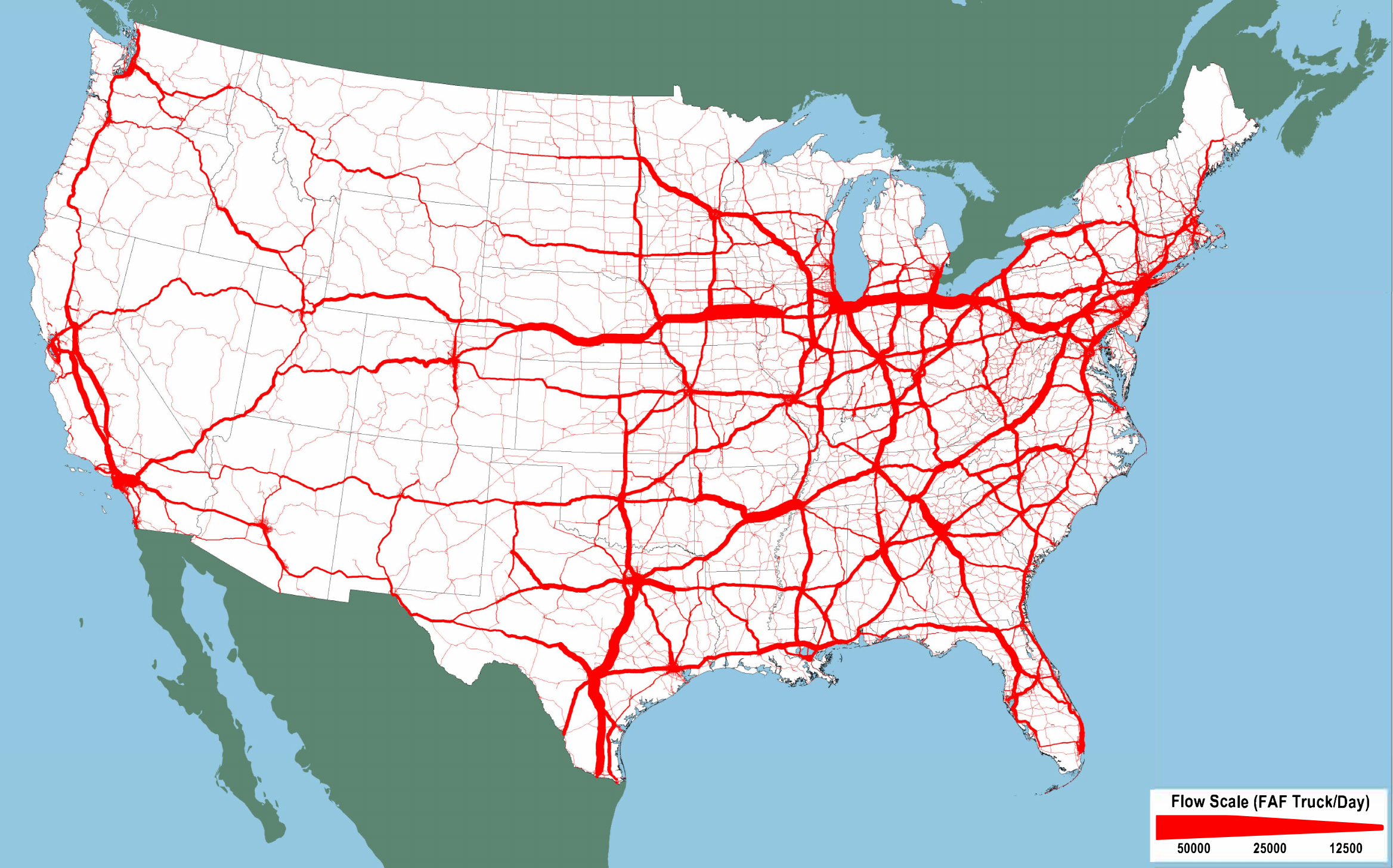}
    \vspace{-1ex}
    \subcaption{$\;$}
    \label{fig:22region}
    \end{subfigure}
\caption{{(a) Illustration of the US national highway network~\cite{wiki:UShighways}. {(b) Average daily long-haul traffic on the US national highway~\cite{truckFreightMap}. Thickness of the lines indicates the traffic flow intensity.}}}
\end{figure}

\subsection{Proof of Lemma \ref{lem:platoonOnce}}
\label{lem:platoonOnceproof}
\begin{proof}
\begin{figure} [hbt]
    \centering
    \includegraphics[width = \columnwidth]{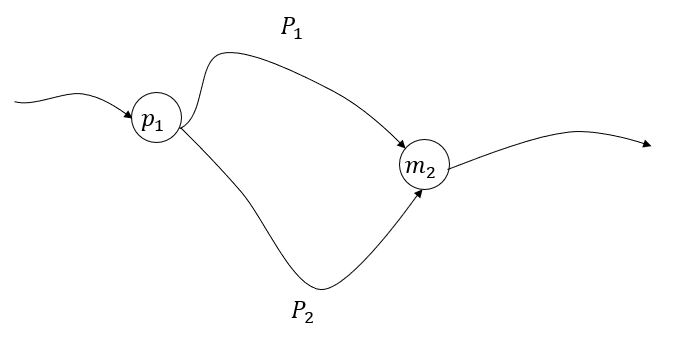}
    \caption{Optimal solution where the two trucks platoon twice.}
    \label{fig:my_platoon_once}
\end{figure}
\rev{For simplicity of presentation, we use $P$ to denote a path, and $C(P)$ to denote the total fuel cost on the path.} As shown in Fig \ref{fig:my_platoon_once}, let $P^*$ be an optimal platoon routing with fuel cost $C(P^*)$ in which there are two sub-paths $P_1$ and $P_2$, and node $p_1$ is the splitting point of the first platoon and node $m_2$ is the merging point of second platoon.
We assume the cost for traversing $P_1$ is $C(P_1)$, and traversing $P_2$ is $C(P_2)$. 
Without loss of generality, we may assume $C(P_2) \leq C(P_1)$. 
In order to re-platoon at $m_2$, the time cost of traversing $P_1$ and $P_2$ for the two trucks should be the same. 
So we can move the truck running on $P_1$ to $P_2$ to form a platoon from $p_1$ to $m_2$. This forms a new solution $P'$.
Then, we have
$$
C(P') = C(P^*) - \eta C(P_1) + \eta C(P_2) \leq C(S),
$$
since $C(P_2) \leq C(P_1)$. Since $P$ is the optimal platooning solution, $C(P^*) \leq C(P')$, therefore, $C(P') = C(P^*)$. This implies that if the two-truck platooning problem is feasible, there exists an optimal platooning solution in which the two trucks platoon only once, and never split and merge for a second time.
\end{proof}

\subsection{Proof of Theorem \ref{Thm:NP}}
\label{Thm:NP_Proof}
\begin{proof} 

    Consider the setting where two trucks share the same
    origin, destination, and pickup and delivery windows. Under this setting, it is clear  that two trucks platooning from the origin to the destination is optimal. The two-truck platooning problem then reduces to the single truck fuel minimization problem~\cite{Deng2016EnergyefficientTT},
    which covers the classical restricted shortest path problem (RSP)
    as a special case. Since RSP is NP-hard, the two-truck platooning
    problem is also NP-hard in general.
    
    Next, we show that the two-truck platooning is \emph{weakly} NP-hard.
    Specifically, we design an FPTAS for the two-truck platooning problem
    that achieves $(1+\epsilon)$ approximation ratio (for any $\epsilon>0$)
    with a time complexity polynomial in the size of the transportation
    network and $1/\epsilon$. The full details are in Appendix \ref{FPTAS}
    and we briefly discussed the idea below. 
    
    We first design an FPTAS under the setting where the merging point
    and splitting point of the platooning are given. Then we obtain an
    FPTAS for the general setting by applying the same FPTAS to all possible
    $\mathcal{O}(N^{2})$ combinations of merging and splitting points. 
    
    More specifically, given the merging point and splitting point, we
    design an FPTAS using the standard rounding and scaling technique.
    The FPTAS consists of two sub-procedures. First, we divide our two-truck
    platooning into five separate RSP problems.
    {For each of the RSP problems with speed planning, we quantize the edge-e fuel-time function $c_e(t_e)$ to be a staircase function, where the number of stairs is determined by $\epsilon$. Meanwhile, the fuel cost is scaled down with a function of $\epsilon$. Using dynamic programming, we enumerate cost bounded minimum-travel-time paths till the deadline is satisfied. Note that the time complexity is dominated by the minimum cost and the cost coordination in the five RSP problems, which is $\mathcal{O}(N^4/\epsilon^4)$.
    }
    Then we combine solutions for the five RSP problems to meet individual
    deadlines and minimize total fuel cost. It can be solved by enumerating
    {all possible $\mathcal{O}(N^2)$ combinations of merging and splitting points. The overall time complexity is thus $\mathcal{O}(N^{6}/\epsilon^{4})$}.
    The existence of an FPTAS implies that the two-truck platooning problem is NP-hard (but) in the weak sense \cite{DBLP:Approx_alg}. 
\end{proof}

\subsection{Checking Feasibility of Two-truck Platooning Problem} \label{subsec:checkFeasi}

{By exploiting Lem.~\ref{lem:platoonOnce}, we can efficiently check the feasibility of our problem using Alg.~\ref{alg:feasCheck}.}

\begin{algorithm2e}[t]
\SetAlgoLined
    \caption{Feasibility Check Algorithm of Two-truck Platooning Problem~$\eqref{equ:objective}$}
    \label{alg:feasCheck}
    Set FEASIBLE\_FLAG=\textrm{FALSE}.\\
    \For{merging node $v_m\in V$}{
        \For{splitting node $v_s\in V$}{
            Set the driving time on edge $e$ to be $t_e^{lb}$ and find the fastest sub-path $p_1$ from $s_1$ to $v_m$, the fastest sub-path $p_2$ from $s_2$ to $v_m$, the fastest sub-path $p_3$ from $v_m$ to $v_s$, the fastest sub-path $p_4$ from $v_s$ to $d_1$, the fastest sub-path $p_5$ from $v_s$ to $d_2$.\\
            Set $\tau_i\xleftarrow{}\sum_{e\in p_i}t_e^{lb},\forall i\in\{1,2,3,4,5\}$\\
            \If{$t_d^1+\tau_1+\tau_3+\tau_4\leq t_a^1$ and $t_d^2+\tau_2+\tau_3+\tau_4\leq t_a^1$ and $t_d^1+\tau_1+\tau_3+\tau_5\leq t_a^2$ and $t_d^2+\tau_2+\tau_3+\tau_5\leq t_a^2$}{
                FEASIBLE\_FLAG $\xleftarrow{}$\textrm{TRUE}\\
                \textbf{return} {FEASIBLE\_FLAG}\\
            }
        }
    }
    \textbf{return} FEASIBLE\_FLAG
\end{algorithm2e}

{To find the shortest path, we can use the classical Dijkstra's algorithm~\cite{wiki:DijkAlg}, which has a worst-case time complexity of
$\mathcal{O}\left(M+N\log N\right)$~\cite{barbehenn1998note}.
Therefore, the worst case time complexity of Alg.~\ref{alg:feasCheck} is
$\mathcal{O}\left(N^2\left(M+N\log N\right)\right)$.
}

\subsection{Proof to Theorem \ref{Thm_ILPGap}}
\label{proof:ILPGap}
\begin{proof}
{The idea of the proof is simple. For any given optimal LP solution, we show that it can be represented as a convex combination of ILP solutions.
Therefore, the main content of this proof is how to find the convex integer decomposition for an optimal LP solution.
In this proof, we use the fact that the minimum fractional flow of an LP solution is unsplittable, and can be represented as a fraction of an integer platooning solution.
The main difficulty is to find the merging and splitting points for the minimum fractional flow in the LP solution.
We use depth first search to find the merging and splitting points corresponding to the minimum fractional flow, whose time complexity is at most $\mathcal{O}(M)$.
Then we remove the current minimum fractional flow in the LP solution, and do the {same thing} for the remaining flow.
{Eventually,} we can decompose the whole LP solution into a convex combination of integer platooning solutions at most $\mathcal{O}(M)$ {steps of search-and-removal}. {Therefore}, for an optimal LP solution, we can recover an integer solution {by selecting the most fuel saving solution in the process} in {at most} $\mathcal{O}(M^2)$ {time}. 
}

{We first prove Lem.~\ref{lem:int_combination}, which will later be used to prove Thm.~\ref{Thm_ILPGap}.

\begin{lemma}
$\forall\left(\vec{x}^*,\vec{y}^*,\vec{z}^*\right)$ that is an optimal solution to the relaxed LP problem of the ILP shown in~\eqref{eq:prob_ppo_ILP}, $\exists{\theta_{k}\in[0,1]}$,
$k\in \mathbb{N}_K$ satisfying $\sum\limits_{k\in \mathbb{N}_K}\theta_{k}=1$, such that,
\begin{align}
  \left(\vec{x}^*,\vec{y}^*,\vec{z}^*\right)=\sum\limits_{k\in\mathbb{N}_K} \theta_{k}\left(\vec{x}^{k}, \mathbf{1}_{v=v_{i_k}},\mathbf{1}_{v=v_{j_k}}\right),
\end{align}
where $\left(\vec{x}^k, \mathbf{1}_{v=v_{i_k}},\mathbf{1}_{v=v_{j_k}}\right)$ is an integer solution to the original ILP problem with merging point $v_{i_k}$ and splitting point $v_{j_k}$.
\label{lem:int_combination}
\end{lemma}
Here, we give a constructive proof to Lem.~\ref{lem:int_combination}. We use $x^{*,i}_e,e\in E, i\in \mathbb{N}_5$($x^{k,i}_e,e\in E, i\in \mathbb{N}_5$ resp.) to denote the variable corresponding to edge $e$ and the sub-path $i$ in the solution $\vec{x}^*$($\vec{x}^k$ resp.). We apply Alg.~\ref{alg:constrt_sol} to construct $\theta_k$ and their corresponding solutions. We iteratively select the minimum positive $x_e^i$ as the coefficient $\theta_k$. To construct $\theta_k$'s \revt{corresponding} integer platooning solution, we execute $6$ searching procedures step by step. If $1=\tilde{i}\defeq\argmin\limits_{i\in\mathbb{N}_5}\min\limits_{e\in E, x_e^i>0}x_e^i$, the step-by-step searching procedure is shown as Alg.~\ref{alg:construct_coefficient_sol}, where we invoke $\mathrm{SEARCH}$ procedure to construct one sub-path. For $\tilde{i}$ other than $1$, we can execute the searching procedure in the same way except for some notation and order rearrangements.   

To show that the construction procedure works, we need to show that the procedure terminates in $\mathcal{O}(M^2)$ time and when it terminates, the output $\theta_k$ and corresponding solutions are what we desire. In every step of the while loop in Alg.~\ref{alg:constrt_sol}, there is at least one positive $x_e^i$ set to 0. Since there are at most $\mathcal{O}(M)$ positive $x_e^i$, the while loop can be executed for at most $\mathcal{O}(M)$ steps. Within every step of the while loop, a searching procedure is executed. In every step of the while loop of Alg.~\ref{alg:search_sol}, one edge is appended into the edge set $p$ and when there are no new edges to be added, the procedure stops, so Alg.~\ref{alg:search_sol} takes time $\mathcal{O}(M)$. To sum up, the construction procedure takes time $\mathcal{O}(M^2)$. And by the construction process, it naturally holds,
\begin{equation}
  \left(\vec{x}^*,\vec{y}^*,\vec{z}^*\right)=\sum\limits_{k\in\mathbb{N}_K} \theta_{k}\left(\vec{x}^{k}, \mathbf{1}_{v=v_{i_k}},\mathbf{1}_{v=v_{j_k}}\right),
\end{equation}
}
\begin{algorithm2e}[t]
\SetAlgoLined
\caption{Construct coefficients $\theta_k$ and their corresponding integer solutions.}
\label{alg:constrt_sol}
       Set $k\xleftarrow{}1$\\
       Set $x_e^i\xleftarrow{}x_e^{*,i},\forall e\in E, i\in\mathbb{N}_5$
      \While{true}{
      \If{$x_e^i=0,\forall e\in E,i\in[5]$}{
       $\mathrm{break}$} 
       $\tilde{e},\tilde{i} \xleftarrow{} \argmin\limits_{e\in E, i\in[5], x_e^{i}>0}{x_e^{i}}$\\
        Set $\theta_k\xleftarrow{}x_{\tilde{e}}^{\tilde{i}}$ and construct it's corresponding solution $\left(\vec{x}^{k}, \mathbf{1}_{v=v_{i_k}},\mathbf{1}_{v=v_{j_k}}\right)$\\
       Set $x^{i}_e\xleftarrow{} x^{i}_e-\theta_kx^{k,i}_e, y_{v_{i_k}}\xleftarrow{}y_{v_{i_k}}-\theta_k, z_{v_{j_k}}\xleftarrow{}z_{v_{j_k}}-\theta_k$\\
          $k$++
      }
\end{algorithm2e}

\begin{algorithm2e}[htbp]
\SetAlgoLined
    \caption{Obtaining $\theta_k$ and its corresponding solution with $\tilde{i}=1$.}
    \label{alg:construct_coefficient_sol}
            $p_1, v_1 \xleftarrow{} \mathrm{SEARCH}(\tilde{e},\tilde{i}, \mathrm{'BACKWARD'})$\\
            $p_1', v_2 \xleftarrow{} \mathrm{SEARCH}(\tilde{e}, \tilde{i}, \mathrm{'FORWARD'})$\\
            $p_1 \xleftarrow{} p_1\cup p_1'$\\
            Take $\tilde{e}_2\in\textbf{In}(v_2)$ satisfying $x_{\tilde{e}_2}^2\geq x_{\tilde{e}}^1$\\
            $p_2, v_3\xleftarrow{}\mathrm{SEARCH}(\tilde{e}_2,2, \mathrm{'BACKWARD'})$\\
            Take $\tilde{e}_3\in\textbf{Out}(v_2)$ satisfying $x_{\tilde{e}_3}^3\geq x_{\tilde{e}}^1$\\
            $p_3, v_4\xleftarrow{}\mathrm{SEARCH}(\tilde{e}_3,3, \mathrm{'FORWARD'})$\\
            Take $\tilde{e}_4\in\textbf{Out}(v_4)$ satisfying $x_{\tilde{e}_4}^4\geq x_{\tilde{e}}^1$\\
            $p_4, v_5\xleftarrow{}\mathrm{SEARCH}(\tilde{e}_4,4, \mathrm{'FORWARD'})$\\
            Take $\tilde{e}_5\in\textbf{Out}(v_4)$ satisfying $x_{\tilde{e}_5}^5\geq x_{\tilde{e}}^1$\\
            $p_5, v_6\xleftarrow{}\mathrm{SEARCH}(\tilde{e}_5,5, \mathrm{'FORWARD'})$\\
            Set $\theta_k\xleftarrow{}x_{\tilde{e}}^1$\\
            Set $x^{k,i}_{e}\xleftarrow{}\mathbf{1}_{e\in p_i},e\in E,i\in \mathbb{N}_5$ and $v_{i_k}\xleftarrow{}v_2, v_{j_k}\xleftarrow{}v_4$\\
\end{algorithm2e}

\begin{algorithm2e}[h]
\SetAlgoLined
    \caption{$\mathrm{SEARCH}$($\tilde{e}$, $\tilde{i}$, $s$)}
    \label{alg:search_sol}
    \While{true}{
        Initialize the forward path's edge set $p\xleftarrow{} \left\{\tilde{e}\right\}$\\
        \While{true}{
            \eIf{$s$='FORWARD'}{
                $v\xleftarrow{}\textbf{head}(\tilde{e})$\\
                $\tilde{E}\xleftarrow{}\textbf{Out}(v)$\\
            }{
                $v\xleftarrow{}\textbf{tail}(\tilde{e})$\\
                $\tilde{E}\xleftarrow{}\textbf{Out}(v)$\\
                
             }
             \eIf{$\exists e'\in \tilde{E}$ such that $x_{e'}^{i}\geq x_{\tilde{e}}^{i}$ and ${e'}\notin p$}{
                $p\xleftarrow{} p\cup\{e'\}$\\
                $\tilde{e}\xleftarrow{}e'$\\
             }{
                \textbf{break}
             }
        }
    }
    Return$(p, v)$ 
\end{algorithm2e}

We use $G(\vec{x})$ to denote the objective function. Suppose the optimal solution of the relaxed LP is 
$\left(\vec{x}^*,\vec{y}^*,\vec{z}^*\right)=\left(\left(x_e^{*,1},x_e^{*,2},x_e^{*,3},x_e^{*,4},x_e^{*,5}\right)_{e\in E},(y_v^*)_{v\in V},(z_v^*)_{v\in V}\right)$. And we define $\vec{x}^{*,i}\coloneqq\left(x_e^{*,i}\right)_{e\in E}$.

We define $\vec{x}^{*,i}$($\vec{x}^{k,i},\forall k\in\mathbb{N}_K$ resp.) as $(x^{*,i}_e)$($x^{k,i}_e,\forall k\in\mathbb{N}_K$ resp.), $\forall e\in E, i\in\mathbb{N}_5$. By Lem.~\ref{lem:int_combination}, $\left(\vec{x}^*,\vec{y}^*,\vec{z}^*\right)$ can be decomposed into {convex combination of} some basic platooning solutions.
\begin{equation*}
    \begin{aligned}
    \vec{x}^*&=\left(\vec{x}^{*,1},\vec{x}^{*,2},\vec{x}^{*,3},\vec{x}^{*,4},\vec{x}^{*,5}\right)\\
    &=\sum_{k\in\mathbb{N}_K}\theta_{k}\left(\vec{x}^{k,1},\vec{x}^{k,2},\vec{x}^{k,3},\vec{x}^{k,4},\vec{x}^{k,5}\right).
\end{aligned}
\end{equation*}
Since the objective function is linear in $\vec{x}$, we have 
\begin{equation*}
    \begin{aligned}
    G(\vec{x}^*) &= G\left(\vec{x}^{*,1},\vec{x}^{*,2},\vec{x}^{*,3},\vec{x}^{*,4},\vec{x}^{*,5}\right)\\
    & =\sum_{k\in\mathbb{N}_K}\theta_{k}G\left(\vec{x}^{k,1},\vec{x}^{k,2},\vec{x}^{k,3},\vec{x}^{k,4},\vec{x}^{k,5}\right). 
\end{aligned}
\end{equation*}

So $\exists k^*$, such that $\theta_{k^*}>0$ and 
\begin{align*}
  G\left(\vec{x}^*\right)\geq G\left(\vec{x}^{k^*,1},\vec{x}^{k^*,2},\vec{x}^{k^*,3},\vec{x}^{k^*,4},\vec{x}^{k^*,5}\right). 
\end{align*}

We fix merging point and splitting point to be ${v}_{i_{k^*}}$ and ${v}_{j_{k^*}}$. By the linearity of objective function $G$ and optimality of $\vec{x}^*$, we have that $\vec{x}^{k^*,1}$ {represents} the shortest path with respect to weight $(w_{{e}}^1)_{e\in E}$ from $s_1$ to $v_{i_{k^*}}$, $\vec{x}^{k^*,2}$ {represents} the shortest path with respect to weight $(w_{{e}}^2)_{e\in E}$ from $s_2$ to $v_{i_{k^*}}$, $\vec{x}^{k^*,3}$ {represents} the shortest path with respect to weight $(w_{{e}}^3)_{e\in E}$ from $v_{i_{k^*}}$ to $v_{j_{k^*}}$, $\vec{x}^{k^*,4}$ {represents} the shortest path with respect to weight $(w_{{e}}^4)_{e\in E}$ from $v_{j_{k^*}}$ to $d_1$ and $\vec{x}^{k^*,5}$ {represents} the shortest path with respect to weight $(w_{{e}}^5)_{e\in E}$ from $v_{j_{k^*}}$ to $d_2$. Therefore, $\left(\left(\vec{x}^{k^*,1},\vec{x}^{k^*,2},\vec{x}^{k^*,3},\vec{x}^{k^*,4},\vec{x}^{k^*,5}\right),\left(\textbf{1}_{v=v_{i_{k^*}}}, \textbf{1}_{v=v_{j_{k^*}}}\right)_{v\in V}\right)$ is an ILP solution. Thus we have 
\begin{align*}
    \mathrm{OPT}_{\mathrm{LP}} & =G\left(\vec{x}^{1,*},\vec{x}^{2,*},\vec{x}^{3,*},\vec{x}^{4,*},\vec{x}^{5,*}\right)\\
    & \geq G\left(\vec{x}^{k^*,1},\vec{x}^{k^*,2},\vec{x}^{k^*,3},\vec{x}^{k^*,4},\vec{x}^{k^*,5}\right)\\
    & \geq\mathrm{OPT}_{\mathrm{ILP}}. 
\end{align*}
Thus we show that optimal value of LP is lower bounded by an ILP solution. So $\mathrm{OPT}_{\mathrm{LP}}\geq\mathrm{OPT}_{\mathrm{ILP}}$ and naturally we have $\mathrm{OPT}_{\mathrm{ILP}}\geq\mathrm{OPT}_{\mathrm{LP}}$. Eventually, $\mathrm{OPT}_{\mathrm{ILP}}=\mathrm{OPT}_{\mathrm{LP}}$.

\end{proof}

\subsection{Proof of Theorem~\ref{thm:dualConvRate}}
\begin{proof}
\label{Proof:dualConvRate}
We follow the proof from the discussion in ~\cite[Section~3.3]{boyd2008subgradient}. {The essence of the proof is that, for a given iteration budget $K$, the optimal constant step size $\phi$ is proportional to $1/\sqrt{K}$. Similarly, the optimality gap of the dual function is also proportional to $1/\sqrt{K}$.} Let $\vec{\lambda}^*$ denote the optimal dual variable.We define the norm of $\vec{\lambda}^*$ as $R\triangleq\norm{\vec{\lambda}^*}$ and $H\triangleq 12\left(\sum_{e\in E}t_e^{ub}+\max\left\{T_d^1, T_d^2\right\}-\min\left\{T_s^1, T_s^2\right\}\right)$, it is easy to verify that,
\begin{equation}
    R\geq\norm{\vec{\lambda}[0]-\vec{\lambda}^*}=\norm{\vec{\lambda}^*}=R,
  \label{equ:dual_init_dist}
\end{equation}
since $\vec{\lambda}[0] = \vec{0}$. Besides,
\begin{equation}
  H\geq\norm{\dot{\vec{\lambda}}}. 
\end{equation}
Furthermore, with the dual subgradient update procedure,  we can derive that, 
\begin{align}
  & \norm{\vec{\lambda}[i+1]-\vec{\lambda}^*}^2 =\norm{\vec{\lambda}[i]+\phi\dot{\vec{\lambda}}[i]-\vec{\lambda}^*}^2 \nonumber\\
  &=\norm{\vec{\lambda}[i]-\vec{\lambda}^*}^2+2\phi\dot{\vec{\lambda}}[i]^T\left(\vec{\lambda}[i]-\vec{\lambda}^*\right)+\phi^2\norm{\dot{\vec{\lambda}}[i]}^2\nonumber\\
  &\leq\norm{\vec{\lambda}[i]-\vec{\lambda}^*}^2-2\phi\left(D^*-{D}_i\right)+\phi^2\norm{\dot{\vec{\lambda}}[i]}^2,\label{inq:dual_rec}
\end{align}
where $D_i\triangleq D\left(\vec{\lambda}[i]\right)$. The last inequality follows from the definition of subgradient, which gives $D^*-{D}_i\leq\dot{\vec{\lambda}}[i]^T\left(\vec{\lambda}^*-\vec{\lambda}[i]\right)$.

By applying~$\eqref{inq:dual_rec}$ recursively, we get
\begin{align}
    \norm{\vec{\lambda}[K+1]-\vec{\lambda}^*}^2 &\leq \norm{\vec{\lambda}[1]-\vec{\lambda}^*}^2-2\phi\sum_{i=1}^K\left(D^*-{D}_i\right) \nonumber\\
    & \quad +\phi^2\sum_{i=1}^K\norm{\dot{\vec{\lambda}}[i]}^2.
\end{align}
Using the fact that $\norm{\vec{\lambda}[K+1]-\vec{\lambda}^*}^2\geq0$, $\norm{\vec{\lambda}[1]-\vec{\lambda}^*}^2\leq R$ and $H\geq\norm{\dot{\vec{\lambda}}}$, we have,
\begin{equation}
  2\phi\sum_{i=1}^K\left(D^*-{D}_i\right)\leq R^2+\phi^2KH^2. 
\end{equation}
Since $D^*-{D}_i\geq D^*-\bar{D}_K$, we have
\begin{equation}
  D^*-\bar{D}_K\leq\frac{R^2+\phi^2KH^2}{2\phi K}. 
\end{equation}
Therefore, with step sizes $\phi_{i}=1/\sqrt{{K}}$, $i\in\mathbb{N}_{K}$, there exists a constant $\xi>0$ such that 
\[
    D^{*}-\bar{D}_{K}\leq\xi/\sqrt{K},
\]
{where $\xi$ can be set to be $\frac{R^2+H^2}{2}$.}
\end{proof}

\subsection{Proof to Theorem \ref{Thm_Optgap}}
\label{Thm:PosteriorBound}
\begin{proof}
First, we show the result for case 1. According to the weak duality, any dual function value is a lower bound of the optimal cost $OPT$, namely it holds that
\begin{equation*}
    D(\vec{\lambda}) \leq OPT.
\end{equation*}
For case 1, the stop condition $\dot{\lambda}_j[k^*] = 0$,$\forall j \in \mathbb{N}_4$ implies that $\delta_{i}^{*}(\vec{\lambda}[k^*])=0,\forall i\in \mathbb{N}_4$. The dual function value $D(\vec{\lambda})$ will be
\begin{equation*}
    D(\vec{\lambda}) = c(\bar{P}) + \sum_{j=1}^4 \lambda_j[k^*]\delta_{i}^{*}(\vec{\lambda}[k^*]) = c(\bar{P}) \leq OPT.
\end{equation*}
{where $\lambda_j[k^*]$ is the dual variable when the dual subgradient algorithm stops.}
Since $\bar{P}$ is a feasible function for our two-truck platooning problem, we have $c(\bar{P}) \geq OPT$. Therefore $c(\bar{P}) = OPT$, we will get the optimal solution if the algorithm returns at line $\ref{line:optima_returned}$.

For case 2, by the same reasoning, we know $\bar{P}$ is a feasible solution, $c(\bar{P}) \geq D^*$, where $D^*$ is the optimal dual function value. 
{By weak duality, $D^* \leq OPT$}.
The optimality gap between $\bar{P}$ and the optimal solution $OPT$ will be
\begin{align*}
    c(\bar{P}) - OPT & \leq  c(\bar{P}) - D^*\\
    & \leq c(\bar{P}) - D^*\\
    & = D(\vec{\lambda}) - \sum_{j=1}^4\lambda_j[k^*] \delta_j - D^*\\
    & \leq -\sum_{j=1}^4\lambda_j[k^*] \delta_j
\end{align*}
The optimality gap between our solution $c(\bar{P})$ and the optimal solution $OPT$ is bounded by the primal-dual gap, which is
\begin{equation*}
    B = c(\bar{P}') - D^* = -\sum_{j=1}^4 \lambda_j[k^*]\delta_j.
\end{equation*}
Note that $\bar{P}$ is feasible only when $\delta_j \leq 0$.
\end{proof}

\subsection{FPTAS}
\label{FPTAS}
In this section, we provide an FPTAS for the fixed merging point and splitting point two-truck platooning problem, which is an extension of FPTASes proposed in~\cite{Hassin1992RSP, Deng2016EnergyefficientTT, Liu2020TITS_energy_truck}. Like RSP, we observe that our two-truck platooning problem also satisfies Bellman's principle of optimality, and can be solved by dynamic programming. In fact, when the merging point and splitting point are fixed, our two-truck platooning problem can be divided into interdependent 5 RSP problems. The difficulty for our two-truck platooning problem is that the deadline constraints of the 5 RSP problems are coupled. So our two-truck platooning problem is quite similar to the multi-task problem in~\cite{Liu2020TITS_energy_truck} for which the deadlines constraints are also interrelated. The only difference between our two-truck platooning problem and the multi-task problem in~\cite{Liu2020TITS_energy_truck} is that the deadline constraints are not coupled sequentially. Therefore, compared to the FPTAS proposed in Section 3 of~\cite{Liu2020TITS_energy_truck}, we can design our FPTAS just by changing the coordination of traveling time for each RSP problem. 
\begin{algorithm2e}[t]
\SetAlgoLined   
    \caption{Test(G, L, U, $\epsilon$)}
    \label{FPTAS:Test}
     Set $S \xleftarrow{} \frac{L\epsilon}{N+1}$, $\hat{U} \xleftarrow{} \lfloor{\frac{U}{S}}\rfloor + 5(N + 1)$, $p_{test} \xleftarrow{}$ NULL\\
     \For{$\forall e \in E$, $\forall \hat{c} = 1, 2, ... \hat{U}$}{
        \label{line:scaling_loop}
        $\hat{c}_e^l \xleftarrow{} \lceil c_e(t_e^u)/S \rceil$, $\hat{c}_e^u \xleftarrow{} \lceil c_e(t_e^l)/S \rceil$\\
        $ t_e^i(\hat{c})\xleftarrow{}\left\{
                \begin{aligned}
                t_e^l, & \quad \hat{c} \geq \hat{c}_e^u;\\
                c_e^{-1}(\hat{c}S), & \quad \hat{c}_e^l \leq \hat{c} \geq \hat{c}_e^u;\\
                + \infty,  & \quad \hat{c} < \hat{c}_e^l.
                \end{aligned}
                \right.
        $\\
    }
    Set $g_{i}(v, 0) \xleftarrow{} +\infty$, $\forall v \neq s_i$, $\forall i \in \mathbb{N}_5$\\
    Set $g_{i}(v, 0) \xleftarrow{} 0$, $\forall v = s_i$, $\forall i \in \mathbb{N}_5, \forall \hat{c} = 1, 2,...,\hat{U}, \forall v \in V$\\
    \For{$\forall i \in \mathbb{N}_5, \forall \hat{c} = 1, 2,...,\hat{U}, \forall v \in V$}{
        \label{line:cost_loop}
        $g_i(v, \hat{c}) \xleftarrow{} g_i(v, \hat{c}-1)$\\
        \For{$\forall e = (u,v) \in E$, $\forall \bar{c} = 1,..., \hat{c}$}{
            $g_i(v, \hat{c}) \xleftarrow{} \min \big \{g_i(v, \hat{c}), g_i(u, \hat{c} - \bar{c}) \big \}$\\
        }
    }
    \For{$\bar{c}_1$ = 1,..., $\hat{U}$}{
        \label{line:short_path_loop}
        \For{$\bar{c}_2$ = 1, ..., $\hat{U} - \bar{c}_1$}{
            \For{$\bar{c}_4$ = 1, ..., $\hat{U} - \bar{c}_1 - \bar{c}_2$}{
                \For{$\bar{c}_5 = 1, ..., \hat{U} - \bar{c}_1 - \bar{c}_2 - \bar{c}_4$}{
                    $\bar{c}_3 = \lfloor (\hat{U} - \bar{c}_1 - \bar{c}_2 - \bar{c}_4 - \bar{c}_5)/( 2- 2\eta) \rfloor$\\
                    $T_1 \xleftarrow{} \max \big \{g_1(d_1, \bar{c}_1), g_2(d_2, \bar{c}_2)\big \} + g_3(d_3, \bar{c}_3) + g_4(d_4, \bar{c}_4)$\\
                    $T_2 \xleftarrow{} \max \big \{g_1(d_1, \bar{c}_1), g_2(d_2, \bar{c}_2)\big \} + g_3(d_3, \bar{c}_3) + g_5(d_5, \bar{c}_5)$\\
                    \If{$t_d^1 + T_1 \leq t_a^1$ and $t_d^2 + T_2 \leq t_a^2$}{
                        $p_{test}$ is defined by the solutions for the 5 RSP problems.\\
                    }
                }             
            }
        }
    }
    {\textbf{return} {$p_{test}$}}
\end{algorithm2e}

The essence of the FPTAS is to quantize and round down edge costs while the travel time information is kept. We use dynamic programming to solve the rounded problem exactly by enumerating the cost-bounded minimum-travel-time path. The rounding procedure guarantees the optimal cost for the rounded problem is polynomially bounded of input size, therefore guarantees a polynomial time complexity. We modify the test procedure of the FPTAS in Section 3.2 of \cite{Liu2020TITS_energy_truck}, especially the traveling time coordination for each RSP. Then we incorporate the rounding and scaling procedure of Algorithm 1 in \cite{Liu2020TITS_energy_truck} to design our FPTAS. The test procedure of our FPTAS is elaborated in Alg \ref{FPTAS:Test}. $s_i$ and $d_i$ denote the source and destination nodes for $i$-th RSP problem, where the sequence of the 5 RSP problems is the same as defined in Section 2 Model and problem formulation.

Compared to Algorithm 1 in Section 3 of \cite{Liu2020TITS_energy_truck}, the only difference is that we need to coordinate the 5 RSP problems differently. The traveling time coordination for the 5 RSP problems is presented in the loop of line 15 to line 28 in Alg \ref{FPTAS:Test}. The rounding and scaling procedure in Alg \ref{FPTAS:Test} guarantees a polynomial time complexity. The proof is conceivable from Section III of~\cite{Deng2016EnergyefficientTT} and Section 3 of~\cite{Liu2020TITS_energy_truck}.  With the test procedure in Alg \ref{FPTAS:Test}, we follow the same search structure as in~\cite{Deng2016EnergyefficientTT} and~\cite{Liu2020TITS_energy_truck} to find the near-optimal solution for our two-truck platooning problem. First, we need an initial lower bound $\textbf{LB}$ where $\textbf{LB} \leq OPT$. Obviously, 0 is a trivial lower bound. However, to accelerate the binary search, we use the solution for the unconstrained two-truck platooning problem as the lower bound $\textbf{LB}$, \emph{i.e.}, the 2 trucks can drive most efficiently on every edge. For the upper bound, we use the solution that the 2 trucks drive most inefficiently on every edge. The binary search procedure is presented in Alg \ref{FPTAS:BinarySearch}. Since the proof for our FPTAS is almost the same as previous work in Section 3 of~\cite{Deng2016EnergyefficientTT} and~\cite{Liu2020TITS_energy_truck}, we omit the proof here.

\begin{algorithm2e}[t]
 \SetAlgoLined
 	\caption{An FPTAS}
	\label{FPTAS:BinarySearch}
    Get a lower bound \textbf{LB} and upper bound \textbf{UB} for \textbf{OPT}.\\
    $B_L \xleftarrow{} \textbf{LB}$\\
    $B_U \xleftarrow{} \textbf{UB}$\\
    \While{$B_U/B_L > 2$}
    {
        $B \xleftarrow{} \sqrt{B_L \cdot B_U}$\\
        \eIf{ Test(G, B, B, $\epsilon$) = NULL}{
            $B_L \xleftarrow{} B$\\
        }{
            $B_U \xleftarrow{} B$\\
        }
    }
    FPTAS $\xleftarrow{}$ Test(G, $B_L$, $2B_U$, $\epsilon$)
\end{algorithm2e}

There are three loops in Algorithms \ref{FPTAS:Test}. The first loop which starts from line \ref{line:scaling_loop} has a time complexity of $\mathcal{O}\left(M\hat{U}\right)$, and the second loop starts from line \ref{line:cost_loop} has a time complexity of $\mathcal{O}\left(M\hat{U}^2\right)$. The time complexity for the third loop starts from line \ref{line:short_path_loop} is $\mathcal{O}\left(\hat{U}^4\right)$. Since $\hat{U}$ is $\mathcal{O}\left(\frac{ \textbf{UB}(N+1)}{\textbf{LB} \epsilon}\right)$, so the time complexity for Algorithm \ref{FPTAS:Test} is the larger one between $\mathcal{O}\left(M\frac{ \textbf{UB}^2(N+1)^2}{\textbf{LB}^2 \epsilon^2}\right)$ and $\mathcal{O}\left(\frac{ \textbf{UB}^4(N+1)^4}{\textbf{LB}^4 \epsilon^4}\right)$. Algorithm \ref{FPTAS:BinarySearch} has a time complexity of $\mathcal{O}\left(\log(\frac{\textbf{UB}}{\textbf{LB}})\right)$. Note that it's easy to find a polynomial bounded upper bound $\textbf{UB}$ for our two-truck platooning problem, like we can simply use the single truck path and speed planning solution in~\cite{Deng2016EnergyefficientTT}. Therefore, Algorithm \ref{FPTAS:Test} and \ref{FPTAS:BinarySearch} are an FPTAS for our two-truck platooning problem.

\begin{IEEEbiography}[{\includegraphics[width=1in,height=1.25in,clip,keepaspectratio]{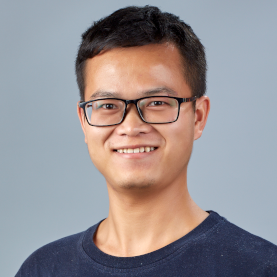}}]{Wenjie Xu} is currently a doctoral student at {\'E}cole polytechnique f{\'e}d{\'e}rale de Lausanne (EPFL). He received his MPhil degree in Information Engineering from The Chinese University of Hong Kong, and B.E. degree in Electronic Engineering from Tsinghua University in 2018. His research interests lie in the interface of optimization, control and machine learning, with applications to building control and intelligent transportation.
\end{IEEEbiography}

\begin{IEEEbiography}[{\includegraphics[width=1in,height=1.25in,clip,keepaspectratio]{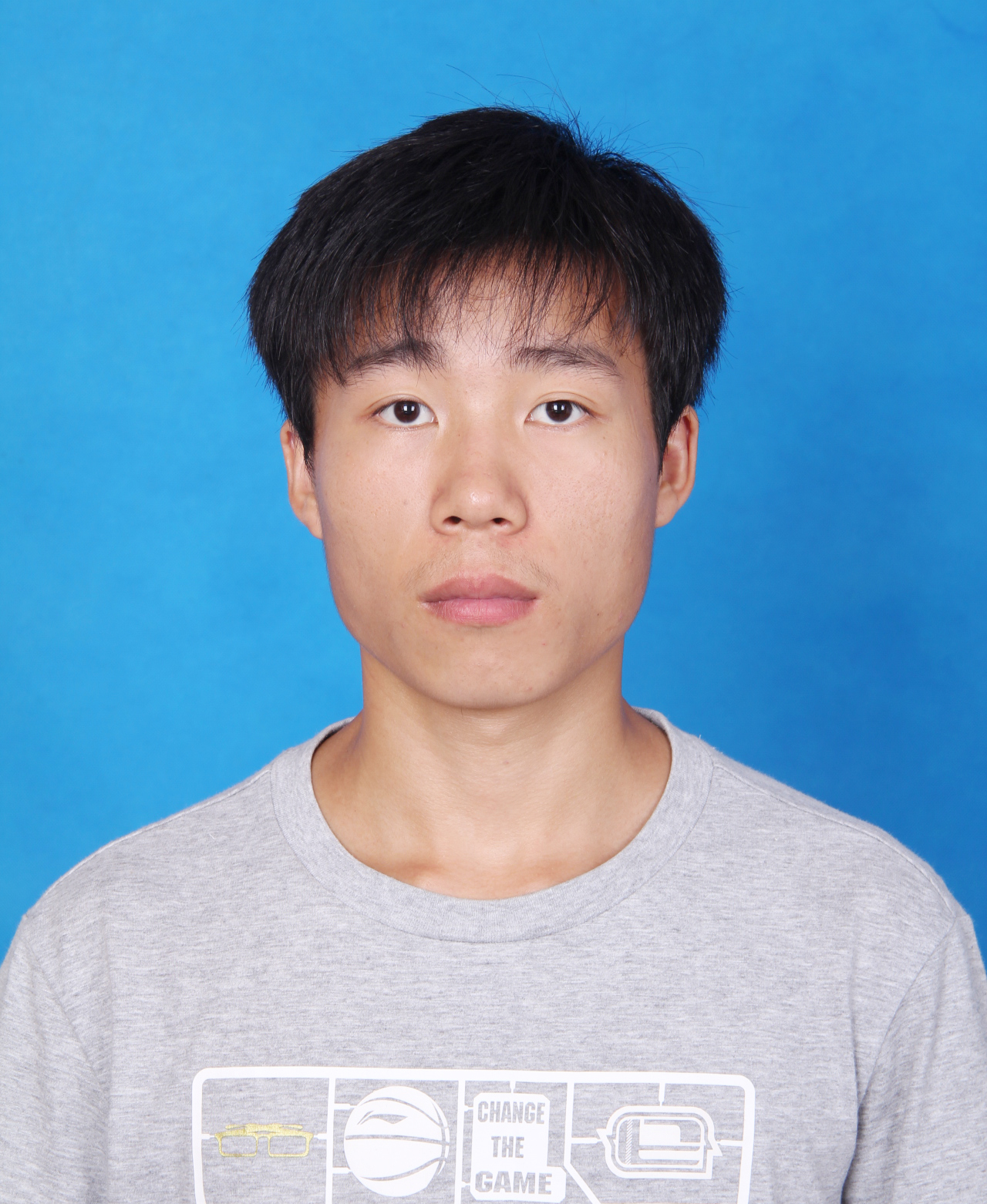}}]{Titing Cui} is currently a doctoral student at University of Pittsburgh. He received his M.S. degree in Applied and Computational Mathematics from KTH Royal Institute of Technology in 2019, and B.S. degree in Mathematics and Applied Mathematics from Zhejiang University in 2017. His research interests include optimization and algorithm design in the fields of the intelligent transportation system, integer programming, revenue management and pricing.
\end{IEEEbiography}

\begin{IEEEbiography}[{\includegraphics[width=1in,height=1.25in,clip,keepaspectratio]{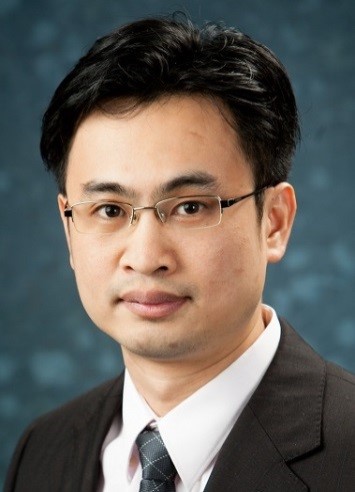}}]{Minghua Chen} (S’04 M’06 SM’ 13 F’22) received his B.Eng. and M.S. degrees from the Department of Electronic Engineering at Tsinghua University. He received his Ph.D. degree from the Department of Electrical Engineering and Computer Sciences at University of California Berkeley. He is currently a Professor of School of Data Science, City University of Hong Kong. Minghua received the Eli Jury award from UC Berkeley in 2007 (presented to a graduate student or recent alumnus for outstanding achievement in the area of Systems, Communications, Control, or Signal Processing) and The Chinese University of Hong Kong Young Researcher Award in 2013. He also received IEEE ICME Best Paper Award in 2009, IEEE Transactions on Multimedia Prize Paper Award in 2009, ACM Multimedia Best Paper Award in 2012, and IEEE INFOCOM Best Poster Award in 2021. He receives the ACM Recognition of Service Award in 2017 and 2020 for the service contribution to the research community. He is currently a Senior Editor for IEEE Systems Journal, an Area Editor of ACM SIGEnergy Energy Informatics Review, and an Award Chair of ACM SIGEnergy. Minghua’s recent research interests include online optimization and algorithms, machine learning in power systems, intelligent transportation systems, distributed optimization, and delay-critical networked systems. He is an ACM Distinguished Scientist and IEEE Fellow.
\end{IEEEbiography}

\end{document}